\PassOptionsToPackage{dvipsnames}{xcolor}
\documentclass[longbibliography,nofootinbib,aps,prx,superscriptaddress,twocolumn]{revtex4-2}

\usepackage{amsmath,amsfonts,amssymb,amsthm,dcolumn,bm,qcircuit,appendix,mathtools,thmtools,thm-restate,mathrsfs,pgfplots,soul,lipsum,graphicx,braket,enumitem,ragged2e}
\usetikzlibrary{calc}
\usepackage{booktabs}
\usepackage{titlesec}
\usepackage{tikz,tkz-euclide,tikz-3dplot}
\usetikzlibrary{positioning}

\usetikzlibrary{matrix,fit,backgrounds,3d,arrows, decorations.pathreplacing,shapes.geometric,3d, calc}

\def\BibTeX{{\rm B\kern-.05em{\sc i\kern-.025em b}\kern-.08em
    T\kern-.1667em\lower.7ex\hbox{E}\kern-.125emX}}
    
\usepackage[ruled,vlined]{algorithm2e}
\usepackage{algorithmic}

\newtheorem{proposition}{Proposition}
\newtheorem{theorem}{Theorem}
\newtheorem{lemma}{Lemma}

\newtheorem{remark}{Remark}
\newtheorem{definition}{Definition}
\newtheorem{method}{Algorithm}

\DeclareMathOperator{\Tr}{Tr}

\newcommand{\Ibb}{\mathbb{I}}

\newcommand{\Rbb}{\mathbb{R}}

\newcommand{\Zbb}{\mathbb{Z}}
\newcommand{\Scal}{\mathcal{S}}
\newcommand{\norm}[1]{\left\lVert#1\right\rVert}

\titleformat{\paragraph}[runin]
{\normalfont\bfseries}
{}
{0em}
{}
\usepackage{xcolor}
\usepackage{hyperref}

\hypersetup{
	colorlinks = true,
	linkcolor = [rgb]{0.70,0.13,0.13},
	citecolor = [rgb]{0.13,0.55,0.13},
	urlcolor = [rgb]{0.25, 0.41, 0.88}}

\begin{document}
\title{Quantum Topological Data Analysis Beyond Betti Numbers: \\ Complexity Hardness $\&$ An Algorithm for Torsion Witness }

\author{Nhat A. Nghiem}
\email{{nhatanh.nghiemvu@stonybrook.edu}}
\email{ {anhnvn@ueh.edu.vn}} 
\affiliation{C. N. Yang Institute for Theoretical Physics, State University of New York at Stony Brook, Stony Brook, NY 11794-3840, USA}
\affiliation{Institute of Applied Mathematics, University of Economics Ho Chi Minh City}

\author{Dominic W. Berry}
\affiliation{School of Mathematical and Physical Sciences, Macquarie University, Sydney, NSW 2109, Australia}

\author{Trung V. Phan}
\thanks{Corresponding author}
\email{tphan@natsci.claremont.edu}
\affiliation{Department of Natural Sciences, Scripps and Pitzer Colleges, \\ Claremont Colleges Consortium, Claremont, CA 91711, United States}

\begin{abstract}
Recent advances have discovered an interesting interplay between quantum computing and topological data analysis (TDA). Most existing work on quantum TDA has focused on Betti numbers, which characterize the intrinsic connectivity and ``holes'' in a dataset. However, homology contains additional information beyond Betti numbers, called \textit{torsion}, where a nontrivial cycle becomes trivial after being repeated a certain number of times, revealing global constraints on how cycles can combine and wrap around one another. 
In this work, we investigate the torsion structure, revealing new limitation and opportunities for quantum speedup. First, given the graph $G$ as input and let $K = \text{Cl}(G)$ be the clique complex, we prove that for a given $r$, determining if the $r$-th homology group $H_r(K,\Zbb)$ contains $p$-torsion, or more formally, an invariant factor with order divisible by a fixed prime $p$, is $\rm NP$-hard. We then discuss consequences and corollaries in a few computationally related problems, such as homological rotor code, Bockstein homomorphism, Smith normal form, lattice saturation and cohomology. Second, letting $P$ be a finite set of prime numbers, we develop a quantum algorithm, which functions as a one-sided torsion witness. For a given $r$, the algorithm outputs two possible outcomes: WITNESS or INCONCLUSIVE. If it is the former, then either  $H_r(K,\Zbb)$ or $H_{r-1}(K,\Zbb)$ has a $p$-torsion for some $p \in P$; otherwise, it is inconclusive about the presence or absence of a $p$-torsion. We then discuss the regime in which a near-quadratic quantum speedup is possible, compared to the corresponding classical algorithm under the same input model. Third, based on the insight of these two results, we show that there is a possibility of superpolynomial quantum speedup in estimating the total number of $p$-torsions to a multiplicative accuracy, and we identify concrete examples that realize this speedup. Our results, especially the $\rm NP$-hardness of determining $p$-torsions, have complemented recent hardness results on the estimation of Betti numbers, completing the complexity-theoretic hardness of integral homology as a whole. At the same time, our work provides another concrete example that realizes large quantum speedup in estimating topologically relevant quantities.


\end{abstract}

\maketitle

\section{Introduction}

{\color{black}Building upon quantum theory, one of the greatest scientific achievements of the twentieth century, quantum computation has emerged as a powerful new computing paradigm.} Considerable effort has been devoted to exploring its potential, leading to many landmark algorithms and results \cite{shor1999polynomial, grover1996fast, deutsch1985quantum, deutsch1992rapid, harrow2009quantum, aharonov2003adiabatic, lloyd1996universal, lloyd2013quantum, lloyd2020quantum, berry2007efficient, berry2012black, berry2014high, berry2015hamiltonian, berry2015simulating, berry2017quantum, childs2017quantum, mitarai2018quantum}. More recently, there has been a surge of interest in applying quantum computation to topological data analysis (TDA), {\color{black}which is} a relatively young but rapidly growing field that uses tools from algebraic topology to study large-scale, complex data. Although TDA offers powerful insights, its methods can be computationally expensive, motivating the search for faster alternatives. Quantum computing, with its ability to manipulate high-dimensional states efficiently, has been viewed as a promising route to overcoming some of these computational challenges.

The interplay between quantum computing and TDA has already produced several interesting results. Lloyd, Garnerone, and Zanardi \cite{lloyd2016quantum} introduced the first quantum algorithm (often called the LGZ algorithm) for estimating Betti numbers, one of the central invariants in TDA. Subsequent works \cite{berry2024analyzing, ubaru2021quantum, mcardle2022streamlined} improved the LGZ framework in various ways. Hayakawa et al.\ \cite{hayakawa2022quantum} proposed a quantum method for estimating persistent Betti numbers, a more refined invariant. The work of \cite{mcardle2022streamlined} proposes a more efficient encoding strategy of simplexes into quantum state, which reduces the qubits required. Berry et al.\ \cite{berry2024analyzing} introduce an optimized version of \cite{lloyd2016quantum}, which yields a more efficient quantum algorithm for estimating Betti numbers. Particularly, they also construct a specific example of complex where the Betti numbers are large, and thus quantum exponential speedup can be possible. In parallel, others have examined the complexity-theoretic aspects of Betti number computation \cite{schmidhuber2022complexity, crichigno2024clique}, revealing that estimating Betti numbers is $\rm NP$-hard and $\rm QMA_1$-hard, thus ruling out the generically efficient solution.

\begin{figure*}
    \centering
    \includegraphics[width=0.8\linewidth]{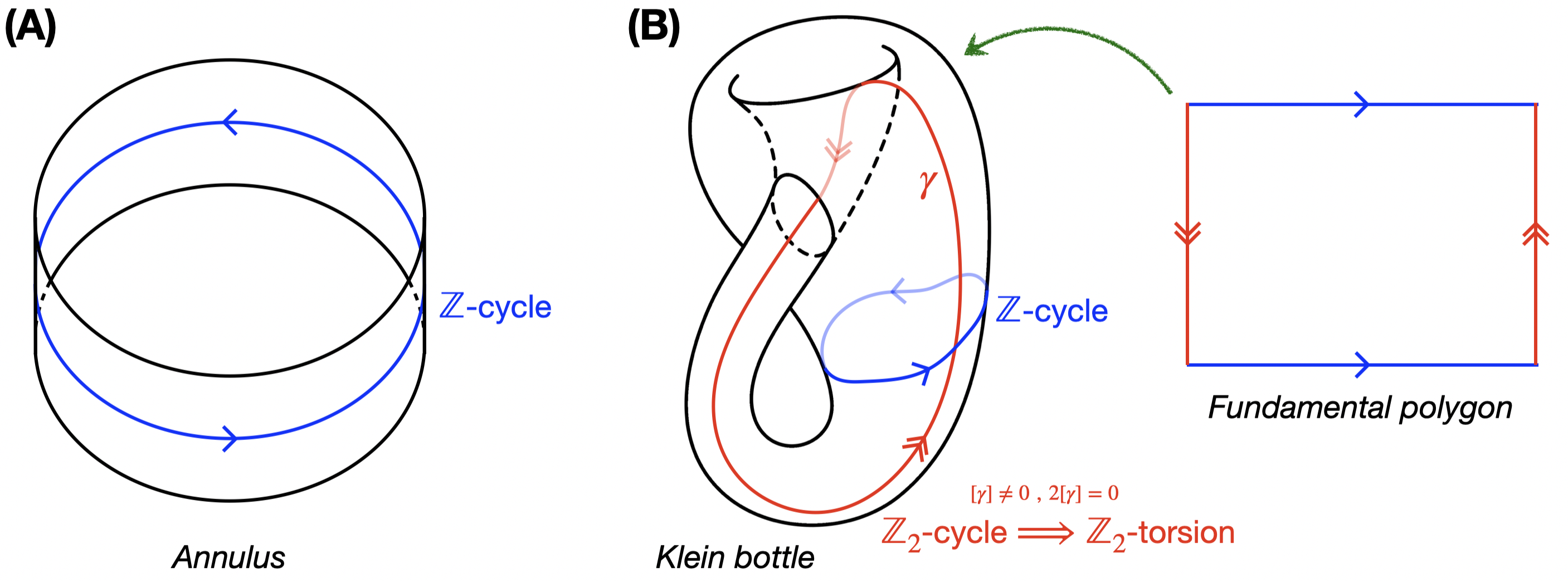}
    \caption{\textbf{Torsion-free and torsion-bearing manifolds.} Torsion is detected in the homology classes of closed cycles. \textbf{(A)} An annulus with a noncontractible $\mathbb{Z}$-cycle. \textbf{(B)} A Klein bottle with a $\mathbb{Z}$-cycle and a $\mathbb{Z}_2$-cycle $\gamma$. Although $\gamma$ is homologically nontrivial, traversing it twice yields a cycle homologous to zero. The fundamental polygon on the right shows the edge identifications responsible for this $\mathbb{Z}_2$-torsion.}
    \label{fig01}
\end{figure*}

While most of these works focus on Betti numbers, they capture only {\color{black}a} part of the {\color{black}available} topological information. For a given simplicial complex, the Betti numbers measure the rank of the free part of the homology group. However, homology also contains another component: torsion\footnote{In this work, torsion refers to homological torsion, not differential geometric torsion \cite{nakahara2018geometry}.}\cite{hatcher2002algebraic}. {\color{black}Intuitively, Betti numbers count independent cycles for which no finite number of copies can form a boundary, whereas torsion detects cycles that are nontrivial individually but become boundaries when finitely many copies are combined.} {\color{black}For example, the annulus and the Klein bottle have the same Betti numbers, but the Klein bottle has $\mathbb{Z}_2$-torsion in its first homology group, whereas the annulus is torsion-free \cite{hatcher2002algebraic} (see Fig.~\ref{fig01}).}

{\color{black}Torsion can provide application-relevant information about finite symmetries, periodic identifications, orientation-reversing structures, and constrained higher-order interactions. Its presence and order can help characterize molecular and crystallographic configuration spaces, neural and higher-order networks, and genomic and RNA organization. Several recent studies have demonstrated this potential, such as $\mathbb{Z}_3$-torsion appears in a directed flag tournaplex of the \textit{C.\ elegans} directed synaptic network \cite{govc2020computing}, while weighted torsion has been used to capture recombination and population structure in viral genomes \cite{barrett2023arithmetic} and higher-order loop interactions in RNA secondary structures \cite{bura2021weighted}.}

Together, the free part and the torsion part give a more complete characterization of a topological space. {\color{black}Nevertheless, most existing quantum TDA algorithms work over real field $\Rbb$, and therefore do not recover the full integral torsion structure.} Readers interested in a more formal discussion of this issue are referred to Appendix~\ref{sec: roleofcoefficients}, which explains how the choice of coefficients can obscure torsion.

Motivated by this gap, our work explores the interplay between quantum computation, classical computation and torsion structure. Specifically, given a simplicial complex $K$, we are interested in seeking what sort of properties that the torsion of the homology group $H_r(K,\Zbb)$ of $K$ might have. Our complexity-hardness proof and algorithm for torsion detection draw on several mathematical tools. The key theoretical foundations are classical results from algebra, including the structure theorem for finitely generated Abelian groups and the universal coefficient theorem, both of which imply that torsion is reflected in the dimensions of certain homology groups. Thus, our strategy reduces to estimating the dimensions of these groups—a task for which we propose a quantum procedure. Our method builds on recent advances in state preparation, block-encoding, and the quantum singular value transformation (QSVT) framework. We show that, with suitable input, a quantum computer can detect torsion quadratically faster than known classical algorithms. In particular, we are able to prove certain complexity-theoretic hardness result related to the torsion structure of $H_r(K,\Zbb)$, for which our proof relies on the well-known Kunneth formula for reduced homology. 

We note that the work of \cite{nghiem2025towards} was the first to provide a quantum algorithm to detect torsion in the input model where an explicit set of simplices is given instead of the oracle that verifies the existence of simplices. The algorithm there is based on the same classical rank sketching method, which is similar to us. However, our quantum algorithm is different in many ways, for example, we use quantum arithmetic more often, meanwhile, the work \cite{nghiem2025towards} relies on the block-encoding of the boundary operator and preparation of certain states based on \cite{zhang2022quantum}. More importantly, the analysis given in \cite{nghiem2025towards} is not completed, as the hardware resource required in that algorithm is proportional to the number of simplices. 


\section{An Overview of Quantum Topological Data Analysis and Our Contribution}
\label{sec: surveyqtda}
\subsection{An overview of existing progress}
Quantum algorithms for topological data analysis began with the work in \cite{lloyd2016quantum}, where the authors considered the problem of estimating the so-called Betti numbers of a given simplicial complex. Their algorithm contains the following key components as well as related assumptions.
\paragraph{Encoding simplexes.} Let the simplicial complex of interest $K$ have $N$ data points, denoted $v_0,v_1,v_2,...,v_N$. Then a $r$-simplex $\sigma_r$ $\in K$ is encoded into a $N$-qubit string of Hamming weight $r+1$ $\ket{\sigma_r}$ as follows: if $v_j \in \sigma_r$ then the value of $j$-th bit is 1. 

\paragraph{Verification oracle $O_r$.} For any $r$ ($1 \leq r \leq N$), there is an oracle $O_r$ that acts as follows:
    \begin{align}
    O_r \ket{0} \ket{\sigma_r} = \begin{cases}
        \ket{1} \ket{\sigma_r} \text{\ if $\sigma_r \in $ K  }\\
        \ket{0} \ket{\sigma_r} \text{ otherwise} 
    \end{cases}
\end{align}
In \cite{lloyd2016quantum}, the authors suggest to use build this oracle from quantum random access memory which holds the pairwise distance between data points. A more realistic construction of this oracle, based on graph knowledge of $K$ plus Toffili gates, is given in \cite{berry2024analyzing}. With the above inputs, LGZ algorithm combines several well-known ones, including Grover's search algorithm (multi-solution version), quantum simulation, and quantum phase estimation, to construct a quantum procedure that estimates the $r$-th normalized Betti number $\frac{\beta_r}{|\Scal_r|}$ (where we remind that $\Scal_r$ is the set of $r$-simplexes in $K$), up to some additive accuracy. 

The LGZ algorithm has ignited many subsequent attempts in many directions. The work in \cite{ubaru2021quantum} proposed to replace Grover's search subroutine within the LGZ algorithm by a method called rejection sampling. In addition, they also proposed replacing the quantum phase estimation subroutine with the stochastic rank estimation method, which was first introduced in \cite{ubaru2016fast,ubaru2017fast}. These two replacements resulted in a certain improvement over the LGZ algorithm, and as claimed in \cite{ubaru2021quantum}, it is very suitable for the near-term era. Recent work \cite{berry2024analyzing} provides an in-depth analysis of quantum advantage in estimating Betti numbers. In particular, they introduced new ways to replace Grover's search algorithm (with Dicke state preparation) and amplitude estimation (using Kaiser window). They also provided specific type of graphs/complexes for which exponential quantum speedup is possible. 

In addition to the two above, the work in \cite{hayakawa2022quantum, mcardle2022streamlined} also features another progress. Instead of estimating Betti numbers (as in \cite{lloyd2016quantum,berry2024analyzing, ubaru2021quantum}), the authors in \cite{hayakawa2022quantum, mcardle2022streamlined} focus on persistent Betti numbers and develop the corresponding quantum algorithms. To elaborate, the Betti numbers capture the number of holes in a given complex. Persistent Betti numbers, on the other hand, capture the number of holes that survive over different \textit{filtration}. To get more insight about filtration, one can imagine that a simplicial complex is built on simplexes, which are again built on data points and pairwise connectivity. Filtration roughly means denser and denser connectivity. 

Aside from the algorithmic aspect, quantum TDA is also attracting attention from the complexity-theoretic aspect. In particular, the work in \cite{schmidhuber2022complexity} has shown that, provided pairwise connectivity, computing Betti numbers is $\# \rm P$-hard, while estimating them is $\rm NP$-hard. A similar result obtained in \cite{crichigno2024clique}, showing that estimating Betti numbers $\rm QMA_1$-hard, while the counting version $\rm \#BQP$-hard. These results have implied a strong barrier on achieving quantum exponential advantage in TDA. However, the work \cite{gyurik2022towards} has shown that a closely related problem to the estimating normalized Betti numbers is classically hard, which leaves a certain hope for quantum TDA. 

\subsection{Motivation}
\label{sec: motivation}

The aforementioned works have featured a significant amount of interest in the interplay between quantum computing and TDA. Here, we make a relevant, yet fundamentally different, progress toward this exciting direction. Instead of seeking the Betti numbers of a given complex, we are curious about the \textit{torsion structure} of such a complex. To elaborate on this concept, we recall that earlier we mentioned that a $r$-chain is formed by taking a (formal) linear combination of $r$-simplexes. The collection of these chains forms the $r$-chain group/space $C_r$. If the coefficients of the linear combination are drawn from real field $\Rbb$, then $C_r$ behaves as a linear vector space, for which the Betti numbers $\beta_r$ can be estimated by probing the spectrum of corresponding Laplacian operators. In particular, we remark that all the works mentioned in the previous section are based on homology \textit{over $\Rbb$} (see also Appendix \ref{sec: reviewofalgebraictopology}). If instead of $\Rbb$, the coefficients are drawn from the integer $\Zbb$, which is a ring, then it leads to a distinct separation. 

First, over $\Zbb$, $C_r$ is no longer a vector space. Rather, it is an Abelian group (or $\Zbb$-module). Second, because of this, $\partial_r$ is no longer a linear mapping between vector spaces, but a homomorphism between groups. Third, as a result, $H_r(K)$ is no longer a space but an Abelian group. For a topological space or chain complex \(X\), the integral homology group admits the decomposition
\begin{align}
\begin{split}
H_r(X;\mathbb Z)
&\cong
\mathbb Z^{\beta_r(X)}
\bigoplus
\operatorname{Tor}H_r(K,\mathbb Z), \\
\operatorname{Tor}H_r(K,\mathbb Z)
&\cong
\bigoplus_{i=1}\mathbb Z_{r_i}.
\end{split}
\end{align}
where $r_i \mid r_{i+1}$ refers to the divisibility relation. The first part $\Zbb^{\beta_r} $ is called \textit{free part}, while the second part $ \left( \bigoplus_i Z_{r_i} \right)$ is \textit{torsion part}. While the free rank \(\beta_r(K)\) is detected by homology in a field of characteristic zero, e.g., $\Rbb$, the finite cyclic factors \(\mathbb Z_{r_i}\) are lost after tensoring with \(\mathbb R\) or \(\mathbb C\). Indeed, there is a well-known property that (see, e.g., Appendix \ref{sec: reviewalgebra} for more details):

$$
H_r(K;\mathbb Z)\otimes \mathbb R
\cong
\mathbb R^{\beta_r(K)},
\qquad
\mathbb Z_{r_i}\otimes\mathbb R=0.
$$
Torsion therefore contains discrete topological information that is fundamentally different from that encoded by Betti numbers. As examples of topological space with torsion, there are real projective space, Klein bottle, Moore space. At the same time, sphere, torus, complex projective space, feature torsion-free space.  We refer the interested readers to the Appendix \ref{sec: reviewalgebra} for a more formal review of abstract algebra as well as a more detailed torsion structure of the aforementioned space. 

As mentioned in the introduction, integral torsion have been applied to the real-world contexts and yielded certain success. Beside these real-world applications, the role of torsion in physics is equally appealing to us.  A particularly direct example arises in homological quantum rotor codes \cite{vuillot2024homological}. In these constructions, the logical degrees of freedom are determined by the integral homology of an underlying chain complex. If the relevant homology group has the form
$$
H_r(K;\mathbb Z)
\cong
\mathbb Z^{\beta_r}
\oplus
\mathbb Z_{r_1}\oplus\cdots\oplus\mathbb Z_{r_T},
$$
the free factors correspond to rotor-like logical degrees of freedom, whereas a torsion factor \(\mathbb Z_{r_i}\) gives rise to a finite-dimensional logical subsystem of dimension \(r_i\). In particular,
$$
\mathbb Z_p\subseteq \operatorname{Tor}H_r( K ;\mathbb Z)
$$
for a prime \(p\) is associated with a finite logical sector carrying a \(p\)-level discrete degree of freedom. More generally, if some invariant factor satisfies \(p\mid r_i\), then the corresponding torsion sector contains an element of order \(p\). Thus the mathematical problem of deciding whether \(H_k(X;\mathbb Z)\) contains \(p\)-torsion can acquire a direct quantum-information interpretation as determining whether the associated homological system supports a nontrivial finite logical sector of order divisible by \(p\).

Torsion also naturally describes discrete charges and gauge sectors in gauge theory and string compactifications \cite{mayrhofer2015discrete}. Let \([\Sigma]\in H_r(K;\mathbb Z)\) be a torsion cycle of order \(p\), so that
$$
p[\Sigma]=0,
\qquad
[\Sigma]\neq 0.
$$
Equivalently, there exists a \((k+1)\)-chain \(B\) satisfying
$$
\partial B=p\Sigma,
$$
even though \(\Sigma\) itself is not a boundary. A physical object, such as a brane, wrapped on \(\Sigma\) can consequently carry a charge that is conserved only modulo \(p\). Schematically,
$$
[\Sigma]\in \mathbb Z_p
\qquad\Longrightarrow\qquad
Q\in\mathbb Z_p,
$$
which is the characteristic structure of a discrete \(\mathbb Z_p\) charge or gauge sector. In compactifications of gauge and string theories, torsion subgroups of the homology or cohomology of the compactification manifold can therefore encode discrete gauge symmetries and finite charge sectors that are invisible to ordinary real-valued homology.

Another closely related example occurs in Abelian Chern--Simons theory on a closed three-manifold \(M\) \cite{guadagnini2014path}. Writing
$$
H_1(M;\mathbb Z)
\cong
\mathbb Z^{b_1(M)}\oplus T,
\qquad
T=\operatorname{Tor}H_1(M;\mathbb Z),
$$
the free and torsion parts lead to qualitatively different sectors of the theory. In particular, the torsion subgroup carries the canonical torsion linking pairing
$$
\lambda:T\times T\longrightarrow \mathbb Q/\mathbb Z,
$$
which enters the Chern--Simons action and the associated topological invariants. For example, when
$$
T\cong \mathbb Z_p,
$$
the theory contains a finite set of topologically distinct torsion sectors labelled by elements of \(\mathbb Z_p\), and their mutual topological information is encoded through the linking form. Thus detecting \(p\)-torsion amounts, in this setting, to determining whether a corresponding finite topological sector exists.

These examples are complementary to the familiar relation between homology and supersymmetric quantum mechanics. For a supersymmetric system whose supercharge realizes a differential \(Q\), zero-energy states satisfy
$$
Q|\psi\rangle=Q^\dagger|\psi\rangle=0,
$$
and, by the Hodge correspondence, the number of zero-energy states in degree \(k\) is
$$
\dim H^k(X;\mathbb C)=\beta_k(X).
$$
Because
$$
\operatorname{Tor}H_k(X;\mathbb Z)\otimes\mathbb C=0,
$$
such a complex-valued formulation detects only the free part of integral homology and is blind to torsion. Torsion therefore captures a different kind of physically relevant topology: rather than counting continuous or free homological degrees of freedom, it characterizes finite logical subsystems, discrete charges, and finite topological sectors.

These physical manifestations provide additional motivation for studying torsion from a computational perspective. Define $p$-torsion as a summand within $ \bigoplus_i Z_{r_i}$ whose order $r_i$ is divisible by $p$. More mathematically precise, we say that $(H_r(K;\mathbb Z))$ contains $p$-torsion if, in its decomposition $H_r(K;\mathbb Z)\cong \mathbb Z^{\beta_r}  \bigoplus_i \mathbb Z_{r_i}$, there exists at least one torsion summand $\mathbb Z_{r_i}$ such that $p\mid r_i$. In this work, we are particularly interested in the following general aspect: 
\begin{center}
    \textbf{For all $r=0,1,2,...,n-1$, what can be the structure of torsion part of $H_r(K, \Zbb)$, or equivalently, what properties can those factors $\{r_i\}$ have ? For example, does $H_r(K, \Zbb)$  contain $p$-torsion ?}
\end{center}
In physical realizations in which these cyclic factors label finite logical, gauge, or topological sectors, the same mathematical question determines whether a corresponding discrete sector is present. Consequently, establishing the computational hardness of \(p\)-torsion detection is not only a statement about the complexity of integral homology, but can also be viewed as evidence that identifying discrete topological structure in such physical systems may itself constitute a nontrivial computational task.

\subsection{ Our contribution and implication}



\paragraph{Complexity-theoretic hardness.}
To this end, we have sufficient information to state our main results. We specifically tackle the question above from both complexity-theoretic perspective and algorithmic perspective. The results are summarized in the following theorems. 
\begin{theorem}[Complexity-theoretic hardness of $p$-torsion]
\label{thm: complexityhardness}
    Given a graph $G$ as input and let $K = \text{Cl}(G)$ be the clique complex built from $G$, for a fixed prime $p$, determining whether $H_r(K,\Zbb)$ contains $p$-torsion  is $\rm NP$-hard.
\end{theorem}
The proof will be given in Section \ref{sec: complexityhardness}. This complexity-theoretic hardness result directly implies that in the context of homological quantum rotor codes, given a homological rotor code specified by a graph $G$, deciding whether the code possess a finite-dimensional logical sector with an order divisible by a fixed prime $p$ is $\rm NP$-hard. This provides another connection between a physical problem and computational complexity-theoretic theory, thus broadening the interplay between quantum computing and TDA as a whole. We recall that before this, the result of \cite{crichigno2024clique} also provides a connection between a classical computational problem and the theory of quantum computational complexity.

\paragraph{ Algorithm for torsion detection.}
Our second result is the quantum algorithm, which can work as a one-sided torsion witness, that can dissect if there is a homology group within $\{H_r(K,\Zbb) \}_{r=0}^{n-1}$ that contains certain $p$-torsion:
\begin{theorem}[One-sided torsion witness]
\label{thm: torsionwitness}
    Let $P$ be a set of prime numbers and $p_{\max} := \max \{ p\}_{p\in P}$. Provided the membership-oracle that can verify the inclusion of simplexes in $K$, there is a quantum algorithm that, with a certain success probability, outputs WITNESS or INCONCLUSIVE. If it is WITNESS, then either $H_r(K,\Zbb)$ or $H_{r-1}(K,\Zbb)$ contains a $p$-torsion for some $p \in P$. Otherwise if it is INCONCLUSIVE, then there is no assertion about the absence or presence of torsion in $H_r(K,\Zbb), H_{r-1}(K,\Zbb)$. The complexity of this algorithm is 
    $$ \mathcal{\tilde{O}}\Big( |P|p^2_{\max}  \sqrt{\binom{n}{r+1}} \text{poly}(n) \Big)$$
    where $\mathcal{\tilde{O}}(.)$ hides the (poly)logarithmic factors. 
\end{theorem}
Under the same input model, we will show that the corresponding classical running time is 
$$ \mathcal{O}\left( |P|  \binom{n}{r+1} \text{poly}(n) \log^3 \frac{1}{\Delta}  \right)$$
Thus, a nearly quadratic quantum speedup in $\binom{n}{r+1}$ is obtained.

\paragraph{ Superpolynomial speedup in estimating $p$-torsions to a multiplicative accuracy.} The above results, particularly the Theorem \ref{thm: complexityhardness}, have revealed that torsion structure is computationally difficult. The algorithm of Theorem \ref{thm: torsionwitness} turns out to yield only quadratic speedup in the best case. If instead of determining whether $p$-torsion exists, we care about the estimation of the total number of $p$-torsions of a given homology group, then the NP-hardness of determining $p$-torsion have ruled out the possibility of efficient estimation. Otherwise, from the estimation, we can infer whether a $p$-torsion exists. This is very similar to the story regarding Betti numbers. It was shown in \cite{schmidhuber2022complexity} that determining whether Betti numbers are nonzero is NP-hard, and thus a generically efficient solution to the estimation of Betti numbers is not possible. However, when targeting an estimation of Betti numbers up to a multiplicative accuracy, then there is a room for a large quantum speedup. In \cite{berry2024analyzing}, it was shown that for a clique complex built from the so-called $r$-th partite graph, an efficient quantum estimation of its $r$-th Betti numbers is possible, yielding a superpolynomial quantum speedup to the best-known classical algorithm.

Inspired by this example, we show that it is possible to achieve a superpolynomial quantum speedup in estimating the total number of $p$-torsions (of a specific homology group) up to a multiplicative accuracy. Our proof is, in fact, an extension of the procedure underlying the proof of Theorem \ref{thm: complexityhardness}. Essentially, in Section \ref{sec: complexityhardness} where we prove Thm.~\ref{thm: complexityhardness}, we use Kunneth formula (\ref{eqn: kunnethformula}) to show the following. Let $K$ be a torsion-free complex, and $P$ be the flag triangulation of $\Rbb \mathbb{P}^2$, then if we build the join complex $K*P$, the (reduced) homology of this join complex is, respectively:
\begin{align}
\widetilde{H}_{r+2}(K*P,\Zbb) \cong (\Zbb/2\Zbb)^{\beta_r(K)}
\end{align}
Otherwise if $P$ is the flag triangulation of Moore space, then we have:
\begin{align}
      \widetilde{H}_{r+2}(K*P,\Zbb) \cong (\Zbb/p\Zbb)^{\beta_r(K)}
\end{align}
where $p$ is a prime number. As can be seen, these homology groups do not have the free part, but only torsion. The total number of $2$-, $p$-torsions of these homology groups is equal to the $r$-th Betti number of $K$, which gives us an alternative means for estimating the number of torsion factors for this specific type of join complex $K*P$. 

\begin{figure}[htbp]
    \centering
    \includegraphics[width=0.8\linewidth]{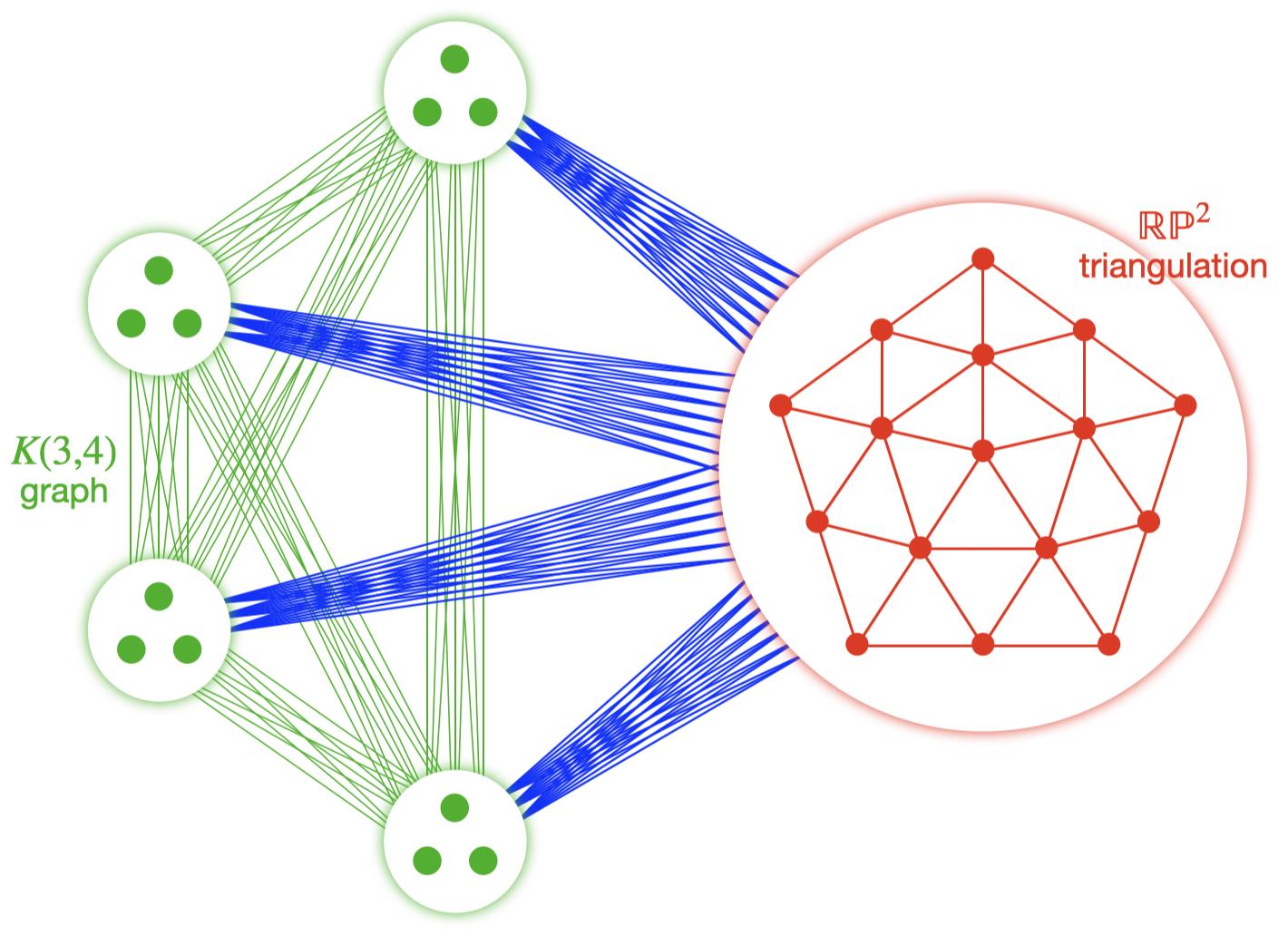}
    \caption{ A concrete instance $\text{Cl}(K(3,4))*P$ (where $P$ now is the flag/clique complex triangulation of $\Rbb\mathbb{P}^2$). Every vertex of the $4$-th partite graph $K(3,4)$ is connected to every vertices of the graph corresponding to $\Rbb\mathbb{P}^2$ triangulation, i.e., join graph. Then the clique complex of the join graph is built, resulting in the join complex $ \text{Cl}(K(3,4))*P$.  }
    \label{fig: joincomplex}
\end{figure}

We recall that it was shown in \cite{berry2024analyzing} that for the clique complex constructed from the so-called $r$-th partite graph $K(m,k)$, under the input assumption of simplex-verification oracle $O_r$ as we discussed earlier, a superpolymial quantum speedup in estimating the $r$-th Betti number is possible. $K(m,r)$ consists of $r$ clusters. Within each cluster, there are $m$ vertices. There are no edges within a cluster, but there are edges connecting vertices from different clusters. By an abuse of notation, denote this clique complex by $K := \text{Cl}(K(m,r))$. According to the previous paragraph, if we build the join complex $K*P$ where $P$ is the flag/clique complex triangulation of the Moore space, then if $K$ is torsion-free, the total number of $p$-torsions in the (reduced) homology group $\widetilde{H}_{r+2}(K*P,\Zbb)$ is equal to the $r$-th Betti number of $K$. We emphasize that this is only possible when $K$ is torsion-free. In the Section \ref{sec: exponentialspeedup}, by iterating the Kunneth formula, we will prove that \textit{$K$, which is the clique complex of the $r$-partite graph, is indeed torsion-free}. Consequently, an estimation of the $r$-th Betti number of $K$ to a multiplicative accuracy implies the estimation with the same multiplicative accuracy for the total number of $p$-torsions of $\widetilde{H}_{r+2}(K*P,\Zbb)$. The superpolynomial quantum speedup in estimating $r$-th Betti number therefore translates over to the estimation of the number of $p$-torsions. Therefore, the join complex $K*P$ serves as a specific family of complexes where a superpolynomial quantum speedup can be realized in estimating torsion-related features, going beyond the Betti numbers.  This adds another example, in addition to the one given in \cite{berry2024analyzing}, demonstrating a large quantum speedup in estimating topologically related quantities. 

To illustrate, we consider $P$ to be the flag triangulation of $\Rbb \mathbb{P}^2$, for simplicity. Figure \ref{fig: joincomplex} features the specific instance of join complex $\text{Cl}(K(3,4)) * P$, which belongs to a broader family of join complexes $ \text{Cl}(K(m,r)) *P$ (where $m,r$ can vary). This particular family of complexes would realize a superpolynomial quantum speedup in estimating the total number of $2$-torsions. 

As a final remark, even though the example above shows a superpolynomial speedup in estimating $p$-torsions, there can still be room for polynomial speedup. For example, in \cite{berry2024analyzing}, the authors identified that for clique complex of random Erdos-Renyi graph, with a suitable inclusion probability, there is a quartic speedup in estimating Betti numbers. Therefore, using the same join complex technique as above,  we can build a family of complexes that realizes quartic speedup in estimating $p$-torsions.

\subsection{Some further consequences and extensions }
Earlier, we pointed out how torsion appears in a few physical contexts like homological rotor code and gauge theory. We also discussed some implication of the hardness of detecting torsion in the context of homological rotor code. Here, we expand the scope of our hardness result to a few related problems. 

A natural consequence of the above reduction is that the hardness result extends well beyond the direct detection of torsion and yields hardness results for several equivalent algebraic and topological formulations. 

\paragraph{Bockstein homomorphism.} First, $p$-torsion can be characterized through the \emph{Bockstein homomorphism} associated with the short exact sequence
\[
0 \longrightarrow \mathbb{Z}
\xrightarrow{\times p}
\mathbb{Z}
\longrightarrow
\mathbb{F}_p
\longrightarrow 0.
\]
The induced connecting map
\[
\beta_p :
H_{r+1}(K,\mathbb{F}_p)
\longrightarrow
H_r(K, \mathbb{Z})
\]
has image equal to the subgroup of $H_r(K, \mathbb{Z})$ annihilated by $p$; hence $H_r(K;\mathbb{Z})$ contains a torsion summand whose order is divisible by $p$ if and only if $\beta_p$ is nonzero. Our reduction therefore also establishes hardness of deciding whether the relevant Bockstein homomorphism is nontrivial. More concretely, we have:
\begin{center}
    \textit{It is $\rm NP$-hard to decide whether the Bockstein homomorphism is nonzero.  }
\end{center}

\paragraph{Smith normal form.} The same phenomenon admits an equivalent formulation in terms of the \emph{Smith normal form} of the boundary matrix $\partial_{r+1}$. If its nonzero invariant factors are $d_1,\ldots,d_r$, then $p$-torsion occurs precisely when
\[
p \mid d_i
\]
for at least one $i$. Equivalently,
\[
\#\{i : p \mid d_i\}
=
\operatorname{rank}_{\mathbb{Q}} \partial_{r+1}
-
\operatorname{rank}_{\mathbb{F}_p} \partial_{r+1},
\]
so detecting a characteristic-dependent rank drop of a succinctly represented boundary matrix is likewise hard. Equivalently, 
\begin{center}
    \textit{It is $\rm NP$-hard to decide whether reducing an integer matrix modulo $p$ increases its rank. }
\end{center}

\paragraph{Lattice saturation.} From the lattice viewpoint, writing
\[
Z_r = \ker \partial_r,
\qquad
B_r = \operatorname{im} \partial_{r+1},
\]
the existence of $p$-torsion in
\[
H_r(K , \mathbb{Z}) = Z_r/B_r
\]
is equivalent to the failure of $B_r$ to be $p$-saturated in $Z_r$, i.e., to the existence of
\[
z \in Z_r \setminus B_r
\]
such that
\[
pz \in B_r.
\]
Thus, the reduction also yields hardness of testing $p$-saturation of simplicial boundary lattices:
\begin{center}
    \textit{Testing $p$-saturation of a succinct simplicial boundary lattice is $\rm NP$-hard.}
\end{center}

\paragraph{Extension to cohomology.}
By the universal coefficient theorem for cohomology,
\[
\operatorname{Tor} H^{r+1}(K , \mathbb{Z})
\cong
\operatorname{Tor} H_r(K , \mathbb{Z}),
\]
so the same construction immediately transfers the hardness result to the detection of $p$-torsion in integral cohomology. We have the following corollary
\begin{center}
    \textit{$p$-torsion detection in integral cohomology is $\rm NP$-hard.}
\end{center}


\section{ Complexity-hardness of detecting $p$-torsion }
\label{sec: complexityhardness}

The proof of hardness of computing torsion factors relies on the following reduction. First, we recall that it was established in \cite{schmidhuber2022complexity} that estimating Betti numbers is $\rm NP$-hard. Our strategy is to construct a complex $K'$ such that its certain torsion factor corresponds to some Betti number of the given complex $K$. Therefore, the $\rm NP$-hardness of the latter implies the hardness of the former. Before proceeding, we note that some of the subsequent graph-theoretic notions and terminologies are summarized in the Appendix \ref{sec: graphtheory}, intended for those readers who are not familiar with. An overview of abstract algebra and related recipes can be found in \ref{sec: reviewalgebra}. Throughout the following, we use $\widetilde{H}$ to denotes the reduced homology.

\paragraph{$\rm NP$-hardness of 2-torsion. }
Let $K$ be the complex of interest, which is built from the underlying graph $G$, i.e., $K = \text{Cl } (G)$. The proof strategy comprises the following ideas:
\begin{itemize}
    \item For a generic graph $G$, it was proved in \cite{schmidhuber2022complexity} that determining whether $\beta_r(K) \equiv \beta_r( \text{Cl } (G) ) > 0 $ is $\rm NP$-hard. This holds even when $G$ is a co-chordal graph. Because $G$ is co-chordal, its complement $\bar{G}$ is chordal, and thus by definition, we have:
    \begin{align}
        \text{Cl } (G) = \text{Ind } (\bar{G})
    \end{align}
    \item In \cite{adamaszek2015note}, it was proved that the independence complex of any chordal graph is homotopy equivalent to either a point or a wedge of spheres. Therefore, for the chordal graph $\bar{G}$, its independence complex is exactly $K \equiv \text{Cl } (G)$, and we have:
    \begin{align}
       K = \text{Cl }(G) \cong S^{d_1} \vee S^{d_2} \vee \cdots \vee S^{d_m}
    \end{align}
    where $d_1,d_2,..., d_m \in \mathbb{Z}$ denotes the dimension of the corresponding spheres. Since the right-hand side is torsion-free, it implies that $\text{Cl } (G)$ is also torsion-free, and thus:
    \begin{align}
        \widetilde{H}_r\big(  K, \Zbb  \big) \cong \Zbb^{\beta_r( K)}
    \end{align}
    \item Next, we consider $\Rbb \mathbb{P}^2$ and take any flag/clique complex triangulation $P$ of $\Rbb \mathbb{P}^2$ \footnote{This is guaranteed to exist, see e.g., \cite{davis2012geometry, hatcher2002algebraic}. More concretely, can first start with any triangulation (which might not be a flag complex) and takes the barycentric subdivision of such triangulation. The resulting complex is a flag complex and homeomorphic to the original complex, thus preserving the homology. }. Let $R$ denotes the graph underlying this triangulation, then almost by definition, we have $P= \text{Cl} (R)$. Because $P$ is a complex corresponding to the triangulation of $\Rbb \mathbb{P}^2$, it is well-known that:
    \begin{align}
        \widetilde{H}_i(P, \Zbb) \cong \begin{cases}
            \Zbb/2\Zbb , i = 1\\
            0, i \neq 1
        \end{cases}
    \end{align}
    \item Next, we form the following graph join:
    \begin{align}
        G'  = G \vee R
    \end{align}
    which is done by keeping all the vertices, edges of $G,R$, while connecting every vertex of  $G$ to every vertex of $R$. As a result, the clique complex of the join graph 
    \begin{align}
        \text{Cl }(G') = \text{Cl } (G) * \text{Cl }(R) = K * P
    \end{align}
    which then means that all the simplices in $K,P$ are hold, while every $0$-simplex of $K$ is now connected to every $0$-simplex of $P$ and include further those simplices formed by these connections. 
    \item The next recipe is the so-called (reduced) Kunneth homology theorem (which can be found in standard text \cite{hatcher2002algebraic, davis2012geometry}):
\begin{widetext}
    \begin{align}
    \widetilde{H}_n(X*Y, \Zbb) \cong \left[ \bigoplus_{p+q=n-1} \widetilde{H}_p(X,\Zbb) \otimes_\Zbb \widetilde{H}_q(Y, \Zbb)  \right] \oplus \left[ \bigoplus_{p+q = n-2} \text{Tor}_1^\Zbb \big( \widetilde{H}_p(X,R), \widetilde{H}_q(Y,R)  \big)  \right]
    \label{eqn: kunnethformula}
\end{align}
\end{widetext}
where $\otimes_\Zbb$ is the tensor product of $\Zbb$-module (see Appendix \ref{sec: reviewalgebra}) and $\rm Tor \left(.  \right) $ is the \textit{torsion product functor} (see further Appendix \ref{sec: roleofcoefficients} for more details). Consequently, we obtain the following: 
    \begin{align}
        \widetilde{H}_{r+2}  ( K * P, \Zbb) \cong \widetilde{H}_r(K,\Zbb) \otimes_\Zbb \Zbb/2\Zbb
        \label{eqn: kunneth}
    \end{align}
    A more detailed description of the Kunneth homology theorem, its reduced form and derivation of this formula shall be given in the Appendix \ref{sec: kunneth}. 
    \item Given that $ \widetilde{H}_r\big(  K, \Zbb  \big) \cong \Zbb^{\beta_r( K)} $, the formula above implies:
    \begin{align}
         \widetilde{H}_{r+2}  ( K * P, \Zbb) \cong  (\Zbb/2\Zbb)^{\beta_r(K)}
    \end{align}
\end{itemize}
Therefore, if $\beta_r(K) > 0$, then $ \widetilde{H}_{r+2}  ( K * P, \Zbb) $  contains torsion of order $2$, or $2$-torsion. Equivalently, if there is an algorithm that can determine whether $2$-torsion exists in polynomial time, then by the polynomial reduction above, this algorithm can also reveal whether $\beta_r(K) > 0$ in polynomial time.  Given that this problem is $\rm NP$-hard, then as a natural consequence, detecting $2$-torsion is also $\rm NP$-hard.

\paragraph{Generalization to $p$-torsion for prime $p$.} This can be done by simply replacing $\Rbb\mathbb{P}^2$ to Moore space $M(\Zbb/p,1)$, as the integral homology of this space is:
\begin{align}
\begin{cases}
    \widetilde{H}_i\big( M(\Zbb/p,1), \Zbb   \big) \cong \Zbb/p\Zbb, i=1 \\
    \widetilde{H}_i\big( M(\Zbb/p,1), \Zbb   \big) = 0 , i\neq 1
\end{cases}
\end{align}
With this change, the whole construction can proceed similar to before, except that $P$ now is the flag triangulation of the Moore space. Because $\Zbb/2\Zbb$ is now replaced by $\Zbb/p\Zbb$, so we end up having:
\begin{align}
      \widetilde{H}_{r+2}  ( K * P, \Zbb) \cong  (\Zbb/p\Zbb)^{\beta_r(K)}
\end{align}
This leads to the same conclusion that detecting $p$-torsion is also $\rm NP$-hard.


\section{Quantum algorithm for detecting torsion}
\label{sec: quantumalgorithm}

\subsection{Mathematical insight}
Let $H_r(Q)$ to denote the $r$-th homology group $H_r$ defined over $Q$, which can be a field (like real field $\Rbb$, complex field $\mathbb{C}$) or a ring (like $\Zbb$). The first recipe we use is the well-known result within algebra, called \textit{universal coefficient theorem}, which shows how the change in the choice of coefficients induces the change in the algebraic structure of the homology group:
\begin{align}
    H_r(Q) \cong H_r(\Zbb) \otimes_{\Zbb} Q  \oplus \rm Tor \left( H_{r-1} (\Zbb), Q  \right)
\end{align}
If we choose $Q = \mathbb{F}_p \equiv \Zbb/p\Zbb$ for prime $p$, then we can turn $\mathbb{F}_p$ into a finite field. As it is a field, the chain group $C_r$ becomes a vector space, and thus the homology group $H_r(\mathbb{F}_p)$ becomes a vector space. In the Appendix \ref{sec: roleofcoefficients} and \ref{sec: keyinsight}, using the above universal coefficient theorem and related properties of torsion product functor, we will prove the following:
\begin{align}
\rm dim \left( H_r(\mathbb{F}_p) \right) = \beta_r + t_r(p) + t_{r-1}(p) 
\end{align}
where $t_r(p), t_{r-1}(p)$ is the total amount of the cyclic summands (of $r$-th homology group/space and $(r-1)$-th homology group/space, respectively) of order divisible by a prime $p$. At the same time, if we choose $Q= \Rbb$ as usual, then $H_r(\Rbb) \cong \Zbb^{\beta_r}$ and thus $\dim H_r(\Rbb) = \beta_r$, which is exactly the $r$-th Betti number. We remark that in the finite field case, the dimension of homology groups are generally not equal to the dimension of the kernel of Laplacian. Rather, it is computed by the following formula:
\begin{align}
   \dim H_r(\mathbb{F}_p)  &= |S_r^K| - \text{rank }_{\mathbb{F}_p} \partial_{r}- \text{rank }_{\mathbb{F}_p} \partial_{r+1}
   \label{5}
\end{align}
Therefore, we need to compute, or more precisely, to estimate the rank of $\partial_r,\partial_{r+1}$ in the finite field $\mathbb{F}_p$. In this case, the computation needs to be done modulo $p$. 

In order to dissect torsion, we compute the dimension of $H_r(\mathbb{F}_p)$ over different values of $p$. Since $\beta_r$ is unchanged, then the change in the $ \dim H_r(\mathbb{F}_p) $ implies that either $t_r(p),t_{r-1}(p)$ changes. Things can be simpler by noting that as $|\Scal_r|$ is unchanged over different finite field $\mathbb{F}_p$'s, it suffices to examine whether the sum of the rank of $\partial_r, \partial_{r+1}$ changes. The change thus reveals at least either of $t_r(p),t_{r-1}(p)$ changes for different $p$ and thus being nontrivial. By repeating for all $r=0,1,2,...,n-1$, we can determine if there is at least one homology group among all $\{H_r(K,\Zbb)\}_{r=0}^{n-1}$ contain torsion. We note one thing that if $\dim H_r(\mathbb{F}_p)$ changes, then it imply the change in either $t_r(p),t_{r-1}(p)$ and thus either of them are nontrivial, which means that either $H_r(K,\Zbb), H_{r-1}(K,\Zbb)$ contains nontrivial torsion. However, we remark that the converse is not true. Even when $\dim H_r(\mathbb{F}_p)$ does not change, it does not mean that $H_r(K,\Zbb), H_{r-1}(K,\Zbb)$ are trivial. For example, it can be that both $t_r(p),t_{r-1}(p)$ changes at different p's but their summation remains the same. Therefore, our quantum algorithm below, which is built on this insight, is functioning as a one-sided torsion witness. 

\subsection{Key recipes}
Below, we summarize the key recipes and tools that we would need for our subsequent outline of the quantum algorithm. More details can be found in the corresponding appendices. 
\paragraph{Input access model.} Similar to previous works, we assume the input oracle that can verify the existence, or inclusion of the simplexes:
 \begin{align}
    O_r \ket{0} \ket{\sigma_r} = \begin{cases}
        \ket{1} \ket{\sigma_r} & \text{\ if $\sigma_r \in $ K  }\\
        \ket{0} \ket{\sigma_r} & \text{ otherwise} 
    \end{cases}
\end{align}

\paragraph{Block-encoding.} A unitary $U$ is said to be a $(\alpha, a, \epsilon)$-encoding of a matrix $A$ if 
$$\left[\bra{0}^{\otimes a} \otimes \Ibb\right] U \left[\ket{0}^{\otimes a}\otimes \Ibb\right] = \frac{\Tilde{A}}{\alpha}$$ 
and $\|\Tilde{A}-A \|_{o} \leq \epsilon$ ($\|.\|_o$ denotes operator norm). Equivalently, in the matrix form
$$
U = \begin{pmatrix}
    \frac{\tilde{A}}{\alpha} & \cdot \\
    \cdot & \cdot
\end{pmatrix}.
$$

\paragraph{Block-encoding of the boundary operator $\partial_r$.} The $(n, a, 0)$-block-encoding of $\partial_r$ can be constructed with a quantum circuit using $\mathcal{O}(n )$ 1- and 2-qubit gates and $\mathcal{O}(1)$ calls to $O_r^K$ (where $a=\mathcal{O}(1)$). This is a result of \cite{berry2024analyzing}. From this block-encoding, we can use the method of \cite{camps2020approximate} to obtain the $(n,a+3,0)$-block-encoding of $\partial_r \otimes \Ibb_3$ using $3$ more ancilla qubits.

\paragraph{Preparing the Dicke state $\frac{1}{\sqrt{\binom{n}{r+1}}} \sum_{|\sigma_r| = r+1  } \ket{\sigma_r} \ket{0}$.} This has also been constructed in \cite{berry2024analyzing} using a circuit of gate complexity $\mathcal{O}(r n)$. Ref.~\cite{khattar2025verifiable} also gives a method for preparing Dicke state with the same gate complexity but using much less ancilla qubits. 

\paragraph{Preparing the state $\frac{1}{\sqrt{|\Scal_r|}} \sum_{\sigma_r \in K } \ket{\sigma_r} \ket{1}$}. This is a consequence of the amplitude amplification. Since we have the oracle $O_r$ acting as $ O_r^K \ket{\sigma_r} = \ket{0/1}$ depending on whether $\sigma_r \in K$, the method in \cite{gilyen2019quantum} can be used to amplify the initial state $\frac{1}{\sqrt{\binom{n}{r+1}}} \sum_{|\sigma_r| = r+1  } \ket{\sigma_r} \ket{0} $ to (approximately) the desired state. The complexity of this amplitude amplification is $\mathcal{O}\big( \sqrt{\frac{\binom{n}{r+1}}{|\Scal_r|}} rn  \big)$. However, first we need to use the amplitude estimation to estimate the ratio $\frac{|\Scal_r|}{\binom{n}{r+1}}$, which has the complexity $\mathcal{O}(\frac{1}{\delta'} \log \frac{1}{\eta})$ for a $\delta'$-approximation additive error and $\eta$ is the failure probability. 

For subsequent use, we define $\{ l(\sigma_r) \}$ to be the natural numbers $\in \Zbb$ corresponding to the set of $n$-bits string with Hamming weight $r+1$. $l(\sigma_r)$ has value falling between $0$ and $\binom{n}{r+1}-1$. 
\paragraph{Preparing the state $\frac{1}{\sqrt{|\Scal_r|}} \sum_{\sigma_r \in K, l(\sigma_r) } \ket{\sigma_r} \ket{1}\ket{l(\sigma_r)}$.}  This can be found in Appendix 7.1 of \cite{gunn2019review}. Specifically, they have shown that by leveraging the combinatorial-number-system rank, precomputed Pascal-triangle table and binary search, one can build a quantum circuit, denoted by $U_l$ of complexity $\mathcal{O}(rn)$ that can convert a string of Hamming weight $r+1$ to a natural number, e.g. $\ket{\sigma_r}\ket{0}^{\otimes \log \binom{n}{r+1}} \longrightarrow \ket{\sigma_r}\ket{l(\sigma_r)}$. Therefore, using this circuit on the superposition
$$ \frac{1}{\sqrt{|\Scal_r|}} \sum_{\sigma_r \in K, l(\sigma_r) } \ket{\sigma_r} \ket{1}\ket{0}^{\log \binom{n}{r+1}}  $$
then we can obtain the desired state.

\paragraph{$\epsilon$-biased distribution.} A distribution $\mathscr{D}$ is called $\epsilon$-biased distribution if elements $u_1,u_2,...,u_n \in \mathbb{F}_p$ drawn from $\mathscr{D}$ satisfy the following:
    \begin{align}
    \Big|  \text{Pr}_{u \sim \mathscr{D}} [ u^T y = 0 ]  - \frac{1}{p} \Big| \leq (1- \frac{1}{p})\epsilon
\end{align}
where $u= (u_1,u_2,...,u_n)^T$ and $y \neq \textbf{0}$ with its entries $\in \mathbb{F}_p$. Equivalently:
\begin{align}
    \text{Pr}_{u \sim \mathscr{D}} [ u^T y = 0] \leq \frac{1}{p} + (1- \frac{1}{p})\epsilon
\end{align}

\paragraph{Classical randomized algorithm for estimating rank (Appendix \ref{sec: randomizedrankestimation}).} As mentioned above, we need to find the rank of $\partial_r,\partial_{r+1}$ in $\mathbb{F}_p$. A classical deterministic algorithm would just diagonalize this matrix (e.g., Gaussian elimination modulo $p$) and find the rank in a straightforward manner. However, most existing quantum primitives are based on linear algebra over $\mathbb{C}$, and thus not suitable to handle the finite-field estimation of the rank. Instead, we will quantize the following classical randomized algorithm for rank estimation introduced in \cite{kaltofen1991wiedemann, eberly2017black}. 
\begin{itemize}
    \item Let $t \in \Zbb_+$ and $V,U$ be matrices of size $|\Scal_r| \times t, t \times |\Scal_{r-1}|$ with entries $\in \mathbb{F}_p$ drawn from the $\epsilon$-biased distribution. 
    \item Form the matrix $ M \equiv U\partial_r V$ (modulo $p$). The entry $ M_{ij}$ of this matrix is: 
    \begin{align}
       M_{ij}=  U_i^T \partial_r V_j  (\text{modulo } p)
    \end{align}
    where $U_i$ is the $i$-th column of $U$ and $V_j$ is the $j$-th column of $V$. Let $\text{rank}(\partial_r), \text{rank}(M)$ denote the rank of $\partial_r$ and $M \equiv U \partial_r V$, respectively.
    \item If 
    $$ t \geq \frac{\log \big( 2 (p^{\text{rank}(\partial_r)}-1) \frac{1}{\Delta} \big) }{ \log \big(  \frac{p}{1+(p-1)\epsilon} \big) }  $$
    then the following holds:
    \begin{align}
        \text{Prob} \left[  \text{rank}(M) =  \text{rank}(\partial_r) \right] \geq 1- \Delta
    \end{align}
\end{itemize}
Our subsequent quantum algorithm will estimate the entries of matrix $M$ using quantum techniques and then use a classical computer to diagonalize $M$ (modulo $p$) to find the rank. A more detailed clarification and proof of the correctness of the algorithm is given in Appendix \ref{sec: randomizedrankestimation}. Here we note that this method is most efficient when the matrix of interest has low rank, because the value of $t$, as can be seen above, can be shown to be $> \rm rank(\partial_r)$. If there is a promise about the nontrivial upper bound, say $\alpha$ ($< |\Scal_r|$, because if $\alpha = |\Scal_r|$ then the rank of $\partial_r < |\Scal_r|$ which is trivial), of rank $\partial_r$. Then we can choose $t$ to be:
$$ t \geq \frac{\log \big( 2 (p^{\alpha}-1) \frac{1}{\Delta} \big) }{ \log \big(  \frac{p}{1+(p-1)\epsilon} \big) } $$
If $\alpha \ll |\Scal_r|$ (or equivalently, the rank of $\partial_r$ is $\ll |\Scal_r|$) then this algorithm achieves its best scaling as in this case $t$ can be negligible compared to the dimension. In the worst case where there is no promise on the nontrivial upper bound $\alpha$, then the value of $t$ can be high and this probabilistic algorithm may not offer speedup compared to standard Gaussian elimination. Regardless, this procedure allows us to leverage quantum computing methods to execute, whereas the standard Gaussian elimination (modulo $p$) does not.

\paragraph{Another formula for $M_{ij}$.} There is an important property about $M_{ij}$ that would be useful for our subsequent construction:
\begin{align}
M_{ij} = U_i^T \partial_r V_j = \sum_{\sigma_r \in \Scal_r} \sum_{k=0}^r (-1)^k U_{i, \sigma_r \backslash v_k } V_{\sigma_r, j}
\end{align}
(modulo $p$) which comes from the following expansion: $ V_j = \sum_{\sigma_r \in \Scal_r} V_{\sigma_r,j} \ket{\sigma_r}, U_i = \sum_{\sigma_{r-1} \in \Scal_{r-1}} U_{\sigma_r,i} \ket{\sigma_{r-1}} $, and the property of boundary operator:
\begin{align}
    \partial_r \ket{\sigma_r} = \sum_{k=0, v_k \in \sigma_r}^r (-1)^k \ket{\sigma_r \backslash v_k}
\end{align}

As the next recipe, we need to construct the unitary, called $U_\sigma$, that achieves the following transformation:
\begin{align}
    U_\sigma \ket{\sigma_r}\ket{k} \ket{0}^{\otimes n} = \ket{\sigma_r}\ket{k} \ket{\sigma_r \backslash v_k}
\end{align}
This is given in the Appendix C of \cite{hayakawa2022quantum}, which has gate complexity $\mathcal{O}(n^2)$.

\paragraph{Preparing the entries of $\epsilon$-biased distribution.} As can be seen above, one of the key steps is to estimate the entry $M_{ij}$. As such, we need a means to prepare the state containing the columns $U_i,V_j$'s. The entries of these vectors are drawn from $\epsilon$-biased distribution. To this end, we state the following lemma and defer the proof to the Appendix \ref{sec: samplingfrombiaseddistribution}:
\begin{lemma}[Preparing states with entries drawn from $\epsilon$-biased distribution]
\label{lemma: preparingstateepsilondistribution}
    Let $m = \mathcal{O}\big( \log_p \frac{N}{\epsilon}\big)$ for some $N \in \Zbb_+$ and $\epsilon$ being a constant. Then there is a reversible quantum circuit $U_e$ of complexity 
    $$ \mathcal{O}\big( \log (N) m^2 \log^2 p  \big)$$
    that acts as follows:
    \begin{align}
        U_e \frac{1}{\sqrt{N}} \sum_{l=0}^{N-1} \ket{i} \ket{00...00}  =\frac{1}{\sqrt{N}} \sum_{l=0}^{N-1} \ket{i)} \ket{v_i}
    \end{align}
    where $v_0,v_1,...,v_{N-1}$ are $\lceil \log p\rceil$-bit string drawn from $\epsilon$-biased distribution.
\end{lemma}

\subsection{Quantum Algorithm}
Our key strategy is to use a quantum computer to estimate the entries $M_{ij}$ of $M$, which is a matrix of size $t \times t$. Then we use a classical computer to diagonalize this matrix to find its rank, which is a good estimator of the real rank of $\partial_r$, with high probability. Our quantum algorithm formally proceeds as follows: 
\begin{widetext}
\begin{method}[One-sided quantum torsion witness]
\label{algo: torsionwitness}
    Let $K = \text{Cl} (G)$ be the clique complex built from the graph $G$ over $n$ vertices. Provided the oracle $\{O_r\}_{r=1}^n$ that can verify the existence of $r$'s-simplexes in $K$. Let $P$ be the set of prime numbers.
    \begin{enumerate}
       \item Prepare the state: 
       \begin{align}
            \frac{1}{\sqrt{|\Scal_r|}} \sum_{\sigma_r \in \Scal_r ,l } \ket{\sigma_r} \ket{0}^{\otimes \log (r+1)} \ket1 
       \end{align}
       which is possible by first preparing $  \frac{1}{\sqrt{|\Scal_r|}} \sum_{\sigma_r \in \Scal_r ,l } \ket{\sigma_r}  \ket1  $ then append ancilla qubits $\ket{0}^{\otimes \log (r+1)}$, then swap the qubit $\ket{1}$ to the last position.
       \item Prepare the state:
        \begin{align}
        \frac{1}{\sqrt{r |\Scal_r|}} \sum_{\sigma_r \in \Scal_r ,l } \sum_{k=0}^r \ket{\sigma_r}\ket{k}   \ket1 
    \end{align}
       by applying $H^{\otimes \log (r+1)} $ to the qubits $\ket{0}^{\otimes \log{(r+1)}}$.
 \item Append another ancilla qubits $\ket{0}^{\otimes n}$, followed by applying the unitary $U_\sigma$:
    \begin{align}
        \frac{1}{\sqrt{r |\Scal_r|}} \sum_{\sigma_r \in \Scal_r ,l } \sum_{k=0}^r \ket{\sigma_r}  \ket{k} \ket{\sigma_r \backslash v_k} \ket{1}
    \end{align}
    \item Append ancilla qubits then using $U_l$  with appropriate qubits to obtain:
    \begin{align}
        \frac{1}{\sqrt{r |\Scal_r|}} \sum_{\sigma_r \in \Scal_r ,l } \sum_{k=0}^r \ket{\sigma_r} \ket{k} \ket{\sigma_r \backslash v_k} \ket{1} \ket{l(\sigma_r)} \ket{l (\sigma_r \backslash v_k) }
    \end{align}
    \item Use $U_e$ (with appropriate ancilla qubits; see Lemma \ref{lemma: preparingstateepsilondistribution}) to obtain the state:
    \begin{align}
        \frac{1}{\sqrt{r |\Scal_r|}} \sum_{\sigma_r \in \Scal_r ,l } \sum_{k=0}^r \ket{\sigma_r} \ket{k} \ket{\sigma_r \backslash v_k} \ket{1} \ket{l(\sigma_r)} \ket{l (\sigma_r \backslash v_k) }  \ket{ v_{l(\sigma_r)}} \ket{v_{l (\sigma_r \backslash v_k) } }
    \end{align}
    where we remind that $ v_{l(\sigma_r)} , v_{l (\sigma_r \backslash v_k) }$ are drawn from $\epsilon$-biased distribution, having values between $0$ and $p-1$. 
    Treating $v_{l(\sigma_r)}$ as $V_{\sigma_r,j}$ and $ v_{l(\sigma_r\backslash v_k)}$ as $U_{i, \sigma_r \backslash v_k}$, then the above state is equivalent to:
    \begin{align}
          \frac{1}{\sqrt{r |\Scal_r|}} \sum_{\sigma_r \in \Scal_r ,l } \sum_{k=0}^r \ket{\sigma_r} \ket{k}  \ket{\sigma_r \backslash v_k}\ket1 \ket{l(\sigma_r)} \ket{l (\sigma_r \backslash v_k) } \ket{ V_{\sigma_r,j}} \ket{ U_{i, \sigma_r \backslash v_k}}
    \end{align}
    
    \item Use the quantum arithmetic circuits \cite{cuccaro2004new, vedral1996quantum, rines2018high, ruiz2017quantum} (with appropriate qubits registers) to obtain the following:
    \begin{align}
          \frac{1}{\sqrt{r |\Scal_r|}} \sum_{\sigma_r \in \Scal_r ,l } \sum_{k=0}^r \ket{\sigma_r} \ket{k}  \ket{\sigma_r \backslash v_k}\ket1 \ket{l(\sigma_r)} \ket{l (\sigma_r \backslash v_k) } \ket{ V_{\sigma_r,j}} \ket{ U_{i, \sigma_r \backslash v_k}} \ket{ (-1)^k  V_{\sigma_r,j}U_{i, \sigma_r \backslash v_k} (\text{mod } (p) )}
    \end{align}
     \item Append an extra ancilla $\ket{0}$ and conditionally transform the ancilla based on the register $\ket{  (-1)^k  V_{\sigma_r,j}U_{i, \sigma_r \backslash v_k}}$ to obtain the state \footnote{In practice, one can either opt to compute the trigonometric function $\arccos(.)$ and rotate the ancilla with  $\ket{  (-1)^k  V_{\sigma_r,j}U_{i, \sigma_r \backslash v_k}}$ as the controlled bit; or we can perform the inequality testing \cite{cuccaro2004new} to obtain the desired transformation. The second approach is more convenient and efficient than the first one, requiring a circuit of complexity $\mathcal{O}(\log^2 p)$ because $V_{\sigma_r,j}, U_{i, \sigma_r \backslash v_k}$ are  $\log p$-bit strings, plus that we do not have to compute the trigonometric function. }:
     \begin{align}
          \frac{1}{\sqrt{r |\Scal_r|}} \sum_{\sigma_r \in \Scal_r ,l } \sum_{k=0}^r \ket{\sigma_r} \ket{k}  \ket{\sigma_r \backslash v_k}\ket1 \ket{l(\sigma_r)} \ket{l (\sigma_r \backslash v_k) } \ket{ V_{\sigma_r,j}} \ket{ U_{i, \sigma_r \backslash v_k}} \ket{ (-1)^k  V_{\sigma_r,j}U_{i, \sigma_r \backslash v_k}} \Big( \frac{(-1)^k  V_{\sigma_r,j}U_{i, \sigma_r \backslash v_k}}{p^2}  \ket{0} +  \\ \sqrt{1- \big( \frac{(-1)^k  V_{\sigma_r,j}U_{i, \sigma_r \backslash v_k}}{p^2}\big)^2} \ket{1}  \Big)
    \end{align}
  where we have hided the modulo $p$. 
   \item Performing $U_l^\dagger, U_e^\dagger, U_\sigma^\dagger$ to uncompute everything and discard the corresponding qubits register, we obtain the simplifed state:
    \begin{align}
       \ket{\phi_1} :=   \frac{1}{\sqrt{r |\Scal_r|}} \sum_{\sigma_r \in \Scal_r ,l } \sum_{k=0}^r \ket{\sigma_r} \ket{k}  \ket1   \Big( \frac{(-1)^k  V_{\sigma_r,j}U_{i, \sigma_r \backslash v_k}}{p^2}  \ket{0} +  \sqrt{1- \big(\frac{(-1)^k  V_{\sigma_r,j}U_{i, \sigma_r \backslash v_k}}{p^2} \big)^2} \ket{1}  \Big)
    \end{align}
    \item Next prepare the following state:
    \begin{align}
        \ket{\phi_2} :=  \frac{1}{\sqrt{r |\Scal_r|}} \sum_{\sigma_r \in \Scal_r ,l } \sum_{k=0}^r \ket{\sigma_r} \ket{k}  \ket1   \ket{0}
    \end{align}
   \item Estimate the overlaps $\braket{\phi_1, \phi_2}$, which is 
   $$ \frac{1}{p^2r|\Scal_r|}\sum_{\sigma_r \in \Scal_r ,l } \sum_{k=0}^r (-1)^k  V_{\sigma_r,j}U_{i, \sigma_r \backslash v_k} = \frac{M_{ij}}{p^2r|\Scal_r|} $$ 
With an appropriate precision (will be analyzed below), the estimation above then translates into the estimation of $M_{ij}$ with accuracy $\frac{1}{2}$. We can then round the estimated value to the nearest integer, which is the correct value of $M_{ij}$.
\item Repeat the whole procedure for $i,j=1,2,...,t$ with $
    t \geq \frac{\log \big( 2 (p^{\alpha}-1) \frac{1}{\Delta} \big) }{ \log \big(  \frac{p}{1+(p-1)\epsilon} \big) } $ ($\alpha$ is the upper bound on the rank of $\partial_r$)
and build the matrix $M$, followed by taking modulo $p$ and classically diagonalizing $M$ to (approximately) find its rank. With high probability $1-\Delta$, this rank is the true rank of $\partial_r$ over $\mathbb{F}_p$.
\item Repeat the same algorithm above to find the rank of $\partial_{r+1}$ over $\mathbb{F}_p$. Then find the summation of the rank of $\partial_r,\partial_{r+1}$. 
\item Repeat the algorithm above for all $p \in P$, to find the summation of the rank of $\partial_r,\partial_{r+1}$ over different $\mathbb{F}_p$'s. 
\item If summation of the rank of $\partial_r,\partial_{r+1}$ changes for some $p$, then the algorithm outputs WITNESS, which indicates that either $H_r(K,\Zbb)$ or $H_{r-1}(K,\Zbb)$ contains $p$-torsion for some $p$. Otherwise, output INCONCLUSIVE, which indicates nothing about the presence or absence of $p$-torsion. 
    \end{enumerate}
\end{method}
\end{widetext}
Below, we discuss several aspects of the algorithm above, including its complexity plus justifying its efficiency regime. 
\subsection{Complexity}

Recall that we begin with the state $\frac{1}{\sqrt{\binom{n}{r+1}}} \sum_{ |\sigma_r|= r+1} \ket{\sigma_r} \ket{0}$, then we use amplitude amplification to amplify this state closely to the state 
$ \frac{1}{\sqrt{|\Scal_r| }} \sum_{ |\sigma_r|= r+1} \ket{\sigma_r} \ket{1}  $. Before the amplitude amplification, we need to use amplitude estimation to estimate the ratio $ \frac{|\Scal_r|}{\binom{n}{r+1]}}$, or more precisely, $ \sqrt{ \frac{|\Scal_r|}{\binom{n}{r+1]}}}$. Let $\delta'$ denote the additive error in the estimation. The complexity of estimating $\frac{|\Scal_r|}{\binom{n}{r+1}}$ to an additive error $\delta'$, with a failure probability $\eta$ is $\mathcal{O}\left(  \frac{1}{\delta'} rn \log( \frac{1}{\eta})    \right)$.
However, this estimation only occurs once. 

The complexity of preparing the state $ \frac{1}{\sqrt{|\Scal_r|}} \sum_{\sigma_r \in \Scal_r ,l } \ket{\sigma_r} \ket{0}^{\log (r+1)}\ket1 $ is the sum of the complexity of performing amplitude estimation plus amplitude amplification, which is (ignoring the failure probability as the amplitude estimation only occurs once):
$$ \mathcal{O}\big( \frac{1}{\delta'}  rn  +   \sqrt{\frac{\binom{n}{r+1}}{|\Scal_r|}} rn   \big) $$


In the Step 5 of algorithm above, we need to use Lemma \ref{lemma: preparingstateepsilondistribution}. The value of $N$ in this Lemma is exactly the value $\binom{n}{r+1}$ in our context, so the complexity of Lemma \ref{lemma: preparingstateepsilondistribution} is 
\begin{align}
    \mathcal{O}\left(  \log^2(p) \log^3 \binom{n}{r+1}  \right) =  \mathcal{O}\left(  \log^2(p) r^3 \log^3 (n) \right) 
\end{align}

The unitary $U_\sigma$ in Step 3 has complexity $\mathcal{O}(n^2)$. The $U_l$ in Step 4 has complexity $\mathcal{O}\big( rn  \big)$. $U_e$ in Step 5 has complexity $ \mathcal{O}\left(  \log^2(p) r^3 \log^3 (n) \right)  $. Step 6 and 7 uses arithmetic circuits \cite{cuccaro2004new, vedral1996quantum}, which operates on $\log p$-bit string and thus having $\mathcal{O}(\log^2 p)$ gate complexity. Step 8 uncomputes all the prior steps, so the total complexity in preparing $\ket{\phi_1}$ is 
\begin{align}
    \mathcal{O}\big( \frac{1}{\delta'}  rn + \log^2(p) r^3 \log^3 (n)  +n^2+  \sqrt{\frac{\binom{n}{r+1}}{|\Scal_r|}} rn   \big)
\end{align}
The cost of preparing $\ket{\phi_2}$ is simply the cost of preparing the state $ \frac{1}{\sqrt{|\Scal_r|}} \sum_{\sigma_r \in \Scal_r ,l } \ket{\sigma_r} \ket{0}^{\log (r+1)}\ket1  $, which is:
\begin{align}
   \mathcal{O}\big( \frac{1}{\delta'}  rn  +   \sqrt{\frac{\binom{n}{r+1}}{|\Scal_r|}} rn  \big) 
\end{align}
So the total complexity in estimating $\braket{\phi_1,\phi_2}$ with an (non-accumulative) accuracy $\delta$ is:
\begin{widetext}
\begin{align}
    \mathcal{O}\left( \frac{1}{\delta'}  rn  + \Big(\log^2(p) r^3 \log^3 (n)+ \sqrt{\frac{\binom{n}{r+1}}{|\Scal_r|}} rn \big)  + n^2\Big) \frac{1}{\delta}   \right) 
\end{align}
\end{widetext}
By choosing $\delta' = \delta$, the complexity above can be simplified as:
\begin{align}
       \mathcal{O}\left( \Big( \log^2(p) r^3 \log^3 (n)+ n^2 +  \sqrt{\frac{\binom{n}{r+1}}{|\Scal_r|}}  rn \Big) \frac{1}{\delta}   \right) 
\end{align}

At the end, we also have to use classical computer to diagonalize the matrix $M$ which is of size $t \times t$ so it takes further $\mathcal{O}(t^3)$ classical time, so the complexity is:
\begin{widetext}
    \begin{align}
           \mathcal{O}\left( t^3 + \Big(  \log^2(p) r^3 \log^3 (n)+  n^2 + \sqrt{\frac{\binom{n}{r+1}}{|\Scal_r|}}  rn \Big) t^2 \frac{1}{\delta}   \right) 
    \end{align}
\end{widetext}
We further recall that with 
\begin{align}
    t \geq \frac{\log \big( 2 (p^{\alpha}-1) \frac{1}{\Delta} \big) }{ \log \big(  \frac{p}{1+(p-1)\epsilon} \big) } 
\end{align}
where $\alpha$ is the upper bound on the rank of $\partial_r,\partial_{r+1}$, then the algorithm success with probability $1-\Delta$.  Eventually we need to repeat the whole algorithm for all $p\in P$, so the complexity is:
\begin{widetext}
\begin{align}
           \mathcal{O}\left( |P| t^3 + \sum_{p=1}^{|P|} \left[ \Big(   \log^2(p) r^3 \log^3 (n)+ n^2  \sqrt{\frac{\binom{n}{r+1}}{|\Scal_r|}}  n^2  \Big) t^2   \frac{1}{\delta} \right]  \right) 
    \end{align}
\end{widetext}

We point out a subtlety that, in the algorithm above, we use amplitude estimation to estimate the overlaps $\braket{\phi_1,\phi_2}$. If we target additive accuracy $\delta$ and failure probability $\eta$, then the complexity of this estimation step alone is $\mathcal{O}\left(  \frac{1}{\delta} \log \frac{1}{\eta}\right)$ (not including the gate required to build the components). By choosing an appropriate value of $\delta$, then we can infer the correct integer value of $M_{ij}$. Since we repeat the whole algorithm for $i,j=1,2,...,t$ and all $p \in P$, then the failure probability is $t^2 |P|\eta$. Using an abuse of notation, if we target a total failure probability $\eta$, then we need to scale $\eta \longrightarrow \eta \frac{1}{t^2 |P|}$. So the complexity would need to contain the factor that accounts for this failure probability, which turns out to be:
\begin{widetext}
\begin{align}
           \mathcal{O}\left( |P| t^3 + \sum_{p=1}^{|P|} \left[ \Big(   \log^2(p) r^3 \log^3 (n)+ n^2 +  \sqrt{\frac{\binom{n}{r+1}}{|\Scal_r|}} rn \Big) t^2   \frac{1}{\delta}  \right]  \log \Big(  t^2 |P|\frac{1}{\eta}\Big) \right) 
    \end{align}
Note that $r \leq n$, so we can simplify the complexity above to:
\begin{align}
     \mathcal{\tilde{O}}\left( |P| t^3 + \sum_{p=1}^{|P|} \left[\Big(  r^3 +    \sqrt{\frac{\binom{n}{r+1}}{|\mathcal{S}_r|}}  rn \Big)  t^2   \frac{1}{\delta}  \right]  \log \Big( t^2 |P| \frac{1}{\eta}\Big) \right) 
\end{align}
\end{widetext}
where $\mathcal{\tilde{O}}(.)$ hides the (poly)logarithmic factor. In fact, in the first step where we need to estimate the ratio $ \frac{|\Scal_r|}{\binom{n}{r+1}}$, then there is also a failure probability but this step only occurs once. The classical randomized algorithm for rank sketching also has failure probability, which is included in the value of $t$. 

Before moving to discussion, we recall that similar to the classical counterpart described earlier, our quantum algorithm would achieve the best performance when $t$ is small. This is possible if the rank of $\partial_r$ is small, or is guaranteed to be upper bounded by a small number.

\subsection{Discussion}

\paragraph{The importance of high precision.} In the first step of the algorithm above, the state 
$$\frac{1}{\sqrt{|\Scal_r|}} \sum_{\sigma_r \in \Scal_r ,l } \ket{\sigma_r}\ket1 $$
is realized perfectly. From then, there is no further error induced as all the operations are realized exactly, so the state $\ket{\phi_1}$ is realized without error. Similarly, the state $\ket{\phi_2}$ is also realized without error.  Eventually, we estimate $\frac{M_{ij}}{p^2r|\Scal_r|}$. If we target an $\delta$-error in the  amplitude estimation of $\braket{\phi_1,\phi_2} $, then the total error accumulated is $\delta$, which means that we are estimating $ \frac{1}{p^2 r |\Scal_r| } M_{ij}$ with an error of $\delta$. 

In the case of finite field, the value of entries of $M$ need to be precisely known, otherwise the resulting erroneous matrix $\widetilde{M}$ (where the entries $\widetilde{M}_{ij}$ is a multiplicative-error estimation of $M_{ij}$), after taking modulo $p$, will have a very different rank. There is generally no relation between the rank of $\widetilde{M}$ (modulo $p$) and of $M$ at all, and thus knowing the rank of $\widetilde{M}$ would not guarantee to infer the rank, or even the range of the rank of $M$. 

To handle this situation, we would need to be able to accurately evaluate the value of $M_{ij}$. This can be done by observing that $M_{ij}$ is an integer, so if we know its estimation up to a precision $\frac{1}{2}$, then we can round the estimation to the nearest integer, which produces the desired value. Equivalently, we need to choose:
\begin{align}
   \delta = \frac{1}{2p^2 r |\Scal_r| } 
\end{align}
Replacing this value, we have the gate complexity:
\begin{widetext}
    \begin{align}
           \mathcal{\tilde{O}}\left( |P| t^3 +  \sum_{p=1}^{|P|} \left[  \Big(     r^3  + \sqrt{\frac{\binom{n}{r+1}}{|\Scal_r|}} n^2 \Big)   t^2  r p^2  |\Scal_r|     \right] \log \Big( t^2 |P| \frac{1}{\eta}\Big)  \right) 
    \end{align}
    \label{38}
\end{widetext}
where we have used $\mathcal{\tilde{O}}(.)$ now to hide the logarithmic  factor.

\paragraph{Performance of the corresponding classical algorithm.} The corresponding classical algorithm proceeds similarly to the quantum algorithm above, except that all the computation steps are done classically. Under the same input model, where we are given an oracle that can query the simplexes, classical algorithm would need to use this oracle $O_r$ to enumerate the whole $r$-chain spaces to verify which $r$-simplexes are actually in $K$. This procedure, by default, requires as many steps as the total number of $r$-simplexes $\binom{n}{r+1}$, thus implying a lower bound $\Omega( \binom{n}{r+1}) $. The next step is to build the matrix $M$, where each $M_{ij}$ requires a matrix-vector multiplication. Since the boundary operator $\partial_r$ is $(r+1)$-column-sparse, so the classical complexity for computing $M_{ij}$ is $\mathcal{O}( r |\Scal_r|+ |\Scal_{r-1}| )$. Building $M$ takes further $t^2$ step, and at the end we need to diagonalize the matrix $M$, which takes further $\mathcal{O}(t^3)$ steps. The whole algorithm then needs to be repeated for all $p \in P$. In total, a classical algorithm has the following complexity:
\begin{align}
    \mathcal{O}\Big( |P| t^3+  \binom{n}{r+1}+ |P| t^2(  r |\Scal_r|+ |\Scal_{r-1}| )\Big) 
\end{align}
with a lower bound $\Omega( \binom{n}{r+1})$. Again, this is also probabilistic. With 
\begin{align}
    t \geq \frac{\log \big( 2 (p^{\alpha}-1) \frac{1}{\Delta} \big) }{ \log \big(  \frac{p}{1+(p-1)\epsilon} \big) } 
\end{align}
the algorithm succeeds with probability $1-\Delta$.

\paragraph{Quantum quadratic speedup.} If the value of $|\Scal_r|$, $t$ and $|P|$ grows at most polynomial in $n$, the quantum complexity in Eqn.~\ref{38} is 
$$\mathcal{O}\Big( |P| p^2_{\max} \sqrt{ \binom{n}{r+1}} \text{poly}(n) \log^3 \frac{1}{\Delta} \Big)$$
where we have defined $p_{\max} := \max\{ p\}_{p \in P}$. Compared to the classical complexity above, which is 
$$\mathcal{O}\Big(  \binom{n}{r+1} + |P| \text{poly}(n) \log^3 \frac{1}{\Delta} \Big)$$ 
there is almost a quadratic speed-up in $\binom{n}{r+1}$. 

We remark that in the algorithm above, we estimate the overlaps $\braket{\phi_1,\phi_2}$, which is equal to $ \frac{M_{ij}}{p^2 r |\Scal_r|}$ (we note that because $M_{ij}$ is positive so we do not need to worry about the sign). In fact, this is done by, e.g., combining Hadamard circuit with amplitude estimation. The value being estimated is actually the square of the overlaps, $ \sqrt{\frac{M_{ij}}{p^2 r |\Scal_r|} }$, because this is the amplitude. It is known \cite{brassard2000quantum} that an additive error $\delta$ in the estimation of $ \sqrt{\frac{M_{ij}}{p^2 r |\Scal_r|} }$ can translate into the same additive error estimation of $ \frac{M_{ij}}{p^2 r |\Scal_r|}$ \cite{brassard2000quantum}. However, when $M_{ij} \ll p^2 r|\Scal_r|$, then an additive error $\delta$ in estimating $ \sqrt{\frac{M_{ij}}{p^2 r |\Scal_r|} } $ can induce an additive error $\mathcal{O}(\delta^2)$ in estimating $ \frac{M_{ij}}{p^2 r |\Scal_r|} $. So, in the estimation of the overlaps $\braket{\phi_1,\phi_2}$, or more precisely, $ \sqrt{|\braket{\phi_1,\phi_2}|}$, we just need to choose an error $\sqrt{\delta} = \sqrt{ \frac{1}{6 p^2 r |\Scal_r|}}$. Using this new error we have the following improved complexity:
\begin{widetext}
 \begin{align}
           \mathcal{\tilde{O}}\left( |P| t^3 +  \sum_{p=1}^{|P|} \left[  \Big(     r^3  + \sqrt{\frac{\binom{n}{r+1}}{|\Scal_r|}} n^2 \Big)   t^2  \sqrt{r p^2  |\Scal_r| }    \right] \log \Big( t^2 |P| \frac{1}{\eta}\Big)  \right) 
    \end{align}
    \end{widetext}
Therefore, even when $|\Scal_r|$ is as large as $\binom{n}{r+1}$, which can be exponentially large in $n$ (when $r$ is high), we still have a quadratic speedup compared to the classical algorithm. We again emphasize that this is only possible in the case where all the values $M_{ij} \ll p^2 r|\Scal_r|$.

\section{Possibility of superpolynomial speedup in estimating total number of $p$-torsions }
\label{sec: exponentialspeedup}
The above results, specifically the first on the NP-hardness to determine the $p$ -torsion, imply that a generically efficient solution in finding the total number of $p$-torsions is not possible. This is somewhat reflected in the second result as our quantum torsion witness algorithm only yields quadratic speedup. Otherwise, we would be able to determine if there is a $p$-torsion efficiently, and thus contradicting the first result. We point out a relevant story on Betti numbers. The hardness results of \cite{crichigno2024clique, schmidhuber2022complexity} suggest that computing, even estimating Betti numbers, generally takes exponential complexity. However, when targeting a multiplicative accuracy in the estimation, there is a possibility for a large speedup in the case where the Betti numbers are as large as the corresponding number of simplices. A concrete example of the complex that admits large Betti numbers is given in \cite{berry2024analyzing}. As analyzed in \cite{berry2024analyzing}, a superpolynomial speedup is achieved in estimating Betti number for this kind of complex.

Here, we attempt to achieve the same thing, as we aim to find concrete example for which an estimation of $p$-torsions up to a multiplicative accuracy can provide superpolynomial speedup. The algorithm \ref{algo: torsionwitness} eventually relies on the finding of the finite-field rank, and as we analyzed above, there is seemingly no room for a beyond-quadratic speedup due to the precision required. Furthermore, as the computation is done modulo $p$, the estimation with multiplicative accuracy is not proper in this finite-field setting. In addition, we note that the outcome of this algorithm is the one-sided torsion witness, which tells whether $p$-torsion exists instead of an estimation of total number of $p$-torsions. As such, relying on this algorithm is not a proper route to achieve large speedup in estimating $p$-torsions.

It turns out that our first result on the hardness of determining the existence of $p$-torsion gives us a way to achieve our goal. Recall that our proof begins with a co-chordal graph $G$, then we form the graph join $G' = G \vee R$ where $R$ is the graph underlying the flag triangulation $P$ of $\Rbb \mathbb{P}^2$. Then we build the clique complex of $G' = K*P$, then we employ the reduced Kunneth formula to derive the following relation:
\begin{align}
            \widetilde{H}_{r+2}  ( K * P, \Zbb) &\cong \widetilde{H}_r(K,\Zbb) \otimes_\Zbb \Zbb/2\Zbb \\
            &\cong (\Zbb/2\Zbb)^{\beta_r(K)}
\end{align}
If $R$ is the graph underlying the flag triangulation of the Moore space instead, then we have:
\begin{align}
            \widetilde{H}_{r+2}  ( K * P, \Zbb) 
            &\cong (\Zbb/p\Zbb)^{\beta_r(K)}
\end{align}
Therefore, the total number of $p$-torsions (within the $\widetilde{H}_{r+2}$ reduced homology group) of the join complex $K*P$ is exactly the $r$-th Betti numbers of the source complex $K$. 

This gives us an alternative way to compute, or more precisely, estimate the total number of $p$-torsions for some particular type of complex which is built by taking the join of two complexes $K$ and $P$. If instead of taking $K = Cl(G)$ where $G$ is the co-chordal graph, we take $K$ to be arbitrary torsion-free complex, then we obtain the same formula as above for the (reduced) homology group of $\tilde{H}_{r+2}(K*P,\Zbb)$. The total number of $p$-torsions of this homology group is equal to the $r$-th Betti number $\beta_r(K)$ of $K$. Therefore, if we can efficiently estimate $\beta_r(K)$, then we can efficiently infer the corresponding number of $p$-torsions of $\widetilde{H}_{r+2}(K*P,\Zbb)$. 

The question now is, what kind of torsion-free complex would it be possible for quantum computers to efficiently estimate the Betti numbers ? The only example of complex that admits efficient quantum estimation is given in \cite{berry2024analyzing}. This particular complex is built by taking the clique complex of the so-called $r$-partite complete graph $K(m,r)$ (in \cite{berry2024analyzing} the notation is $K(m,r)$ but we replace $k$ by $r$ to match our notation in this work). To remind, $K(m,r)$ consists of $r$ clusters. Within each cluster, there are $m$ vertices. There are no edges within a cluster, but there are edges connecting vertices from different clusters. As proved in \cite{berry2024analyzing}, the clique complex built from $K(m,r)$ has large Betti numbers $\beta_r = (m-1)^{r+1}$, plus a large spectral gap, which enables a polynomial complexity for the quantum estimation of $\beta_r$ up to a multiplicative accuracy. As also shown in \cite{berry2024analyzing}, for this particular complex, under the simplex-verification oracle assumption, there is a superpolynomial speedup compared to the best-known classical deterministic algorithm. 

If we choose the complex $K$ in the discussion above to be the $\text{Cl} (K(m,r))$, then if this complex is torsion-free, the total number of $p$-torsions within the $\widetilde{H}_{r+2}(K*P,\Zbb)$ of the resulting complex $K*P$ is equal to $\beta_r(K)$. As mentioned above, quantum computer can efficiently estimate this quantity. Consequently, quantum computers can also efficiently estimate the total number of $p$-torsions of the complex $\text{Cl} (K(m,r))*P$. In order to achieve this, we emphasize that $K$ needs to be torsion-free, which is proved below. 
\begin{lemma}
$K := \text{Cl} (K(m,r))$ is torsion-free.
\end{lemma}

\paragraph{Proof.} The proof again relies on the Kunneth formula for integral homology. We treat each cluster as a simplicial complex $D$ with $m$ disconnected vertices. Then the complex $X$ is exactly the join of $k$ such complexes:
\begin{align}
    K = \underbrace{D*D* \cdots * D}_{r \ \text{times}}
\end{align}
Since each $D$ is composed of $m$ disconnected vertices, its reduced homology group is simply:
\begin{align}
    \widetilde{H}_j(D,\Zbb) = \begin{cases}
        \Zbb^{m-1}, j=0 \\
        0 , j\neq 0
    \end{cases}
\end{align}
It is obvious that the torsion part of $\widetilde{H}_j(D,\Zbb)$ is zero, almost by definition. According to the Kunneth formula \ref{eqn: kunnethformula}, the (reduced) homology of the join $D*D$ is:
\begin{widetext}
\begin{align}
    \widetilde{H}_1(D*D,\Zbb) \cong \left[ \widetilde{H}_0(D,\Zbb) \otimes_\Zbb \widetilde{H}_0 (D,\Zbb)   \right] \oplus \left[ \text{Tor} \big(    \widetilde{H}_0(D,\Zbb), \widetilde{H}_{-1}(D,\Zbb) \big) \right]
\end{align}    
\end{widetext}
Because the torsion part of $ \widetilde{H}_0(D,\Zbb)$ is zero, so the term $ \text{Tor} \big(    \widetilde{H}_0(D,\Zbb), \widetilde{H}_0(D,\Zbb) \big) = 0$. So we have:
\begin{align}
     \widetilde{H}_1(D*D,\Zbb) \cong \Zbb^{m-1} \otimes_\Zbb \Zbb^{m-1} \cong \Zbb^{(m-1)^2}
\end{align}
Using Kunneth formula again, we have:
\begin{widetext}
    \begin{align}
        \widetilde{H}_2(D*D*D, \Zbb) \cong \left[\widetilde{H}_1(D*D,\Zbb) \otimes_\Zbb  \widetilde{H}_0(D,\Zbb)  \right] \oplus \left[ \text{Tor} \big(    \widetilde{H}_1(D*D,\Zbb), \widetilde{H}_{-1}(D,\Zbb) \big) \right]
    \end{align}
\end{widetext}
and note that the torsion part of $ \widetilde{H}_1(D*D,\Zbb) $ is also zero, we have:
\begin{widetext}
    \begin{align}
        \widetilde{H}_2(D*D*D, \Zbb) \cong \left[\widetilde{H}_1(D*D,\Zbb) \otimes_\Zbb  \widetilde{H}_0(D,\Zbb)  \right] \cong \Zbb^{(m-1)^2} \otimes_\Zbb \otimes \Zbb^{m-1} \cong \Zbb^{(m-1)^3}
    \end{align}
\end{widetext}
Note that the torsion part of $\widetilde{H}_2(D*D*D, \Zbb)  $ is also zero.
Iterating this procedure by keep applying the Kunneth formula for a total of $r$ times, and  then we obtain:
\begin{align}
    \widetilde{H}_{r} (K,\Zbb) \cong \Zbb^{(m-1)^{r+1}} 
\end{align}
Therefore, at any degree $r$, the torsion part of $ \widetilde{H}_{r}(K,\Zbb)$ is zero, implying that $K$ is torsion-free. The free part of this reduced homology group has rank $(m-1)^{r+1}$, which means that the $r$-th Betti number $\beta_r = (m-1)^{r+1}$. This reproduces the value of the $r$-th Betti number of the clique complex of $K(m,r)$ as proved in \cite{berry2024analyzing}. Our work essentially takes a step further, as we show that such clique complex is also torsion-free. $\blacksquare$

To summary, we have shown that $K := \text{Cl} (K(m,r))$ is torsion-free. Therefore, the complex built by taking the join of $K$ and $P$ (where $P$ is the flag triangulation of the Moore space), $K*P$, would have the reduced homology $\widetilde{H}_{r+2} (K*P,\Zbb)$ to be isomorphic to $(\Zbb/p\Zbb)^{\beta_r(K)}$. So the total number of $p$-torsions, at the degree $r+2$ of this complex, is equal to the $r$-th Betti number of $K$. A superpolynomial advantage in estimating this quantity (up to a multiplicative accuracy) would imply the same advantage for estimating $p$-torsions of $\widetilde{H}_{r+2} (K*P,\Zbb)$.

\section{Conclusion}
\label{sec: conclusion}
In this work, we have explored quantum TDA beyond Betti numbers. We specifically considered the homology group $H_r(K,\Zbb)$ of a given complex $K$ and shown that determining $p$-torsion from this group is $\rm NP$-hard. As torsion is intrinsically related to many physical problems, this implies that the computational problems regarding certain physical quantities, e.g., the finite dimensional logical sector,  are also $\rm NP$-hard. At the same time, we build a quantum algorithm, which is a one-sided torsion witness, that in appropriate scenarios, can determine if the complex $K$ contains $p$-torsion for $p$ belonging to some finite set of prime numbers. Our work has expanded the scope of current quantum TDA, which mostly focuses on the Betti numbers,  to the torsion part (where as Betti numbers belong to the free part). The results have revealed the complexity-theoretic limit to the related problems in computational TDA, showing specifically that integral homology is computationally difficult. Whether our Algorithm \ref{algo: rankestimatingfinitefield} can be improved so that beyond quadratic speedup is possible in detecting torsion presence is left as an open question.

\section*{Acknowledgements}
Part of this work is done when N.A.N. is at Google Quantum AI. We acknowledge the use of ChatGPT, developed by OpenAI in the work. This tool was used to supply proof ideas as well as improving the presentation of the work. This research is partly funded by University of Economics Ho Chi Minh City (UEH), Vietnam. \\

\textit{Note added. } During the final stage of our project, we are aware of recent work \cite{wesolowski2026torsion}, which has a similar result to us regarding the $\rm NP$-hardness of $2$-torsion detection. Our proof of the hardness of the $p$-torsion is more direct using the Moore space.

\bibliography{ref.bib}
\bibliographystyle{unsrt}

\newpage
\appendix
\onecolumngrid

\section{Block-encoding and quantum singular value transformation}
\label{sec: summaryofnecessarytechniques}
We briefly summarize the essential quantum tools used in our algorithm. For conciseness, we highlight only the main results and omit technical details, which are thoroughly covered in~\cite{gilyen2019quantum}. An identical summary is also presented in~\cite{lee2025new}.

\begin{definition}[Block-encoding unitary, see e.g.~\cite{low2017optimal, low2019hamiltonian, gilyen2019quantum}]
\label{def: blockencode} 
Let $A$ be a Hermitian matrix of size $N \times N$ with operator norm $\norm{A} < 1$. A unitary matrix $U$ is said to be an \emph{exact block encoding} of $A$ if
\begin{align}
    U = \begin{pmatrix}
       A & * \\
       * & * \\
    \end{pmatrix},
\end{align}
where the top-left block of $U$ corresponds to $A$. Equivalently, one can write
\begin{equation}
    U = \ket{\mathbf{0}}\bra{\mathbf{0}} \otimes A + (\cdots),    
\end{equation}
where $\ket{\mathbf{0}}$ denotes an ancillary state used for block encoding, and $(\cdots)$ represents the remaining components orthogonal to $\ket{\mathbf{0}}\bra{\mathbf{0}} \otimes A$. If instead $U$ satisfies
\begin{equation}
    U = \ket{\mathbf{0}}\bra{\mathbf{0}} \otimes \tilde{A} + (\cdots),
\end{equation}
for some $\tilde{A}$ such that $\|\tilde{A} - A\| \leq \epsilon$, then $U$ is called an {$\epsilon$-approximate block encoding} of $A$. Furthermore, the action of $U$ on a state $\ket{\mathbf{0}}\ket{\phi}$ is given by
\begin{align}
    \label{eqn: action}
    U \ket{\mathbf{0}}\ket{\phi} = \ket{\mathbf{0}} A\ket{\phi} + \ket{\mathrm{garbage}},
\end{align}
where $\ket{\mathrm{garbage}}$ is a state orthogonal to $\ket{\mathbf{0}}A\ket{\phi}$. The circuit complexity (e.g., depth) of $U$ is referred to as the {complexity of block encoding $A$}.
\end{definition}

Based on~\ref{def: blockencode}, several properties, though immediate, are of particular importance and are listed below.
\begin{remark}[Properties of block-encoding unitary]
The block-encoding framework has the following immediate consequences:
\begin{enumerate}[label=(\roman*)]
    \item Any unitary $U$ is trivially an exact block encoding of itself.
    \item If $U$ is a block encoding of $A$, then so is $\Ibb_m \otimes U$ for any $m \geq 1$.
    \item The identity matrix $\Ibb_m$ can be trivially block encoded, for example, by $\sigma_z \otimes \Ibb_m$.
\end{enumerate}
\end{remark}

Given a set of block-encoded operators, various arithmetic operations can be done with them. Here, we simply introduce some key operations that are especially relevant to our algorithm, focusing on how they are implemented and their time complexity, without going into proofs. For more detailed explanations, see~\cite{gilyen2019quantum, camps2020approximate}.

\begin{lemma}[Informal, product of block-encoded operators, see e.g.~\cite{gilyen2019quantum}]
\label{lemma: product}
    Given unitary block encodings of two matrices $A_1$ and $A_2$, with respective implementation complexities $T_1$ and $T_2$, there exists an efficient procedure for constructing a unitary block encoding of the product $A_1 A_2$ with complexity $T_1 + T_2$.
\end{lemma}

\begin{lemma}[Informal, tensor product of block-encoded operators, see e.g.~{\cite[Theorem 1]{camps2020approximate}}]\label{lemma: tensorproduct}
    Given unitary block-encodings $\{U_i\}_{i=1}^m$ of multiple operators $\{M_i\}_{i=1}^m$ (assumed to be exact), there exists a procedure that constructs a unitary block-encoding of $\bigotimes_{i=1}^m M_i$ using a single application of each $U_i$ and $\mathcal{O}(1)$ SWAP gates.
\end{lemma}

\begin{lemma}[Informal, linear combination of block-encoded operators, see e.g.~{\cite[Theorem 52]{gilyen2019quantum}}]
    Given the unitary block encoding of multiple operators $\{A_i\}_{i=1}^m$. Then, there is a procedure that produces a unitary block encoding operator of $\sum_{i=1}^m \pm (A_i/m) $ in time complexity $\mathcal{O}(m)$, e.g., using the block encoding of each operator $A_i$ a single time. 
    \label{lemma: sumencoding}
\end{lemma}

\begin{lemma}[Informal, Scaling multiplication of block-encoded operators] 
\label{lemma: scale}
    Given a block encoding of some matrix $A$, as in~\ref{def: blockencode}, the block encoding of $A/p$ where $p > 1$ can be prepared with an extra $\mathcal{O}(1)$ cost.
\end{lemma}



\begin{lemma}[Matrix inversion, see e.g.~\cite{gilyen2019quantum, childs2017quantum}]\label{lemma: matrixinversion}
Given a block encoding of some matrix $A$  with operator norm $||A|| \leq 1$ and block-encoding complexity $T_A$, then there is a quantum circuit producing an $\epsilon$-approximated block encoding of ${A^{-1}}/{\kappa}$ where $\kappa$ is the conditional number of $A$. The complexity of this quantum circuit is $\mathcal{O}\left( \kappa T_A \log \left({1}/{\epsilon}\right)\right)$. 
\end{lemma}

\begin{lemma}\label{lemma: amp_amp}[\cite{gilyen2019quantum} Theorem 30]
\label{lemma: amplification}
Let $U$, $\Pi$, $\widetilde{\Pi} \in {\rm End}(\mathcal{H}_U)$ be linear operators on $\mathcal{H}_U$ such that $U$ is a unitary, and $\Pi$, $\widetilde{\Pi}$ are orthogonal projectors. 
Let $\gamma>1$ and $\delta,\epsilon \in (0,\frac{1}{2})$. 
Suppose that $\widetilde{\Pi}U\Pi=W \Sigma V^\dagger=\sum_{i}\varsigma_i\ket{w_i}\bra{v_i}$ is a singular value decomposition. 
Then there is an $m= \mathcal{O} \Big(\frac{\gamma}{\delta}
\log \left(\frac{\gamma}{\epsilon} \right)\Big)$ and an efficiently computable $\Phi\in\mathbb{R}^m$ such that
\begin{equation}
\left(\bra{+}\otimes\widetilde{\Pi}_{\leq\frac{1-\delta}{\gamma}}\right)U_\Phi \left(\ket{+}\otimes\Pi_{\leq\frac{1-\delta}{\gamma}}\right)=\sum_{i\colon\varsigma_i\leq \frac{1-\delta}{\gamma} }\tilde{\varsigma}_i\ket{w_i}\bra{v_i} , \text{ where } \Big|\!\Big|\frac{\tilde{\varsigma}_i}{\gamma\varsigma_i}-1 \Big|\!\Big|\leq \epsilon.
\end{equation}
Moreover, $U_\Phi$ can be implemented using a single ancilla qubit with $m$ uses of $U$ and $U^\dagger$, $m$ uses of C$_\Pi$NOT and $m$ uses of C$_{\widetilde{\Pi}}$NOT gates and $m$ single qubit gates.
Here,
\begin{itemize}
\item C$_\Pi$NOT$:=X \otimes \Pi + I \otimes (I - \Pi)$ and a similar definition for C$_{\widetilde{\Pi}}$NOT; see Definition 2 in \cite{gilyen2019quantum},
\item $U_\Phi$: alternating phase modulation sequence; see Definition 15 in \cite{gilyen2019quantum},
\item $\Pi_{\leq \delta}$, $\widetilde{\Pi}_{\leq \delta}$: singular value threshold projectors; see Definition 24 in \cite{gilyen2019quantum}.
\end{itemize}
\end{lemma}
\begin{lemma}
\label{lemma: qsvt}[\cite{gilyen2019quantum} Theorem 56]
\label{lemma: theorem56}  
Suppose that $U$ is an
$(\alpha, a, \epsilon)$-encoding of a Hermitian matrix $A$. (See Definition 43 of~\cite{gilyen2019quantum} for the definition.)
If $P \in \mathbb{R}[x]$ is a degree-$d$ polynomial satisfying that
\begin{itemize}
\item for all $x \in [-1,1]$: $|P(x)| \leq \frac{1}{2}$,
\end{itemize}
then, there is a quantum circuit $\tilde{U}$, which is an $(1,a+2,4d \sqrt{\frac{\epsilon}{\alpha}})$-encoding of $P(A/\alpha)$ and
consists of $d$ applications of $U$ and $U^\dagger$ gates, a single application of controlled-$U$ and $\mathcal{O}((a+1)d)$
other one- and two-qubit gates.
\end{lemma}


\section{Abstract Algebra}
\label{sec: reviewalgebra}
In this section, we provide an overview of algebra, showing essential concepts to understand the meaning of torsion. As a consequence, we aim to comprehend the role of coefficients and how it can affect the Betti numbers -- a central problem within TDA, as mentioned in the main text. A more detailed introduction to abstract algebra can be found in standard textbook, e.g., \cite{gallian2021contemporary,grillet2007abstract}.

\begin{definition}[Group]
A group $\mathcal{G}$ is a set equipped with an operation, denoted as $*$, with the following so-called group axioms:
\begin{enumerate}
    \item \textbf{Closure:} if $a,b \in \mathcal{G}$ then $a * b \in X$. 
    \item \textbf{Associativity:} $a*(b*c) = (a*b) *c$.
    \item \textbf{Identity: } There is a namely identity element $e$ such that $a*e = e*a = a$. 
    \item \textbf{Inverse: } Every element $a$ has an inverse element $a^{-1}$ s.t. $a a^{-1} = a^{-1} a = e$. 
\end{enumerate}
\end{definition}

A group is said to be Abelian if the group operation between two elements commute, i.e., for any $a,b \in \mathcal{G}$, we have that $a * b = b*a$. Abelian group is very common in many contexts, e.g., any vector space is an Abelian group. 

\begin{definition}[Torsion]
 Given an Abelian group $\mathcal{G}$, an element $ x \in \mathcal{G}$ is a \textbf{torsion element} if there is some positive integer $n$ that makes $x*x*x*\cdots * x $ ($n$ times) an identity element $e$. If there is no such integer $n$ exists, then $x$ is called \textbf{free}. 
\end{definition}
\noindent
\textbf{Examples:} 
\begin{itemize}
    \item In $\mathbb{Z}_6$ (integers mod 6), which we conveniently denote as $\{ 0, 1,2,3,4,5 \}$, then $2$ is torsion element. 
    \item In $\mathbb{Z}$, there is no torsion except 0 -- the identity element. 
\end{itemize}

Given multiple groups, it is possible to ``build'' a larger group from these smaller groups while keeping their respective structure. The procedure is called \textit{direct sum}. 
\begin{definition}[Direct Sum]
    If $A$ and $B$ are two groups, then their direct sum is:
    \begin{align}
        A \oplus B = \{  (a,b) \ |\  a \in A, b \in B \}
    \end{align}
    with the component-wise operation is defined as:
    \begin{align}
        (a_1,b_1) + (a_2, b_2) = (a_1+ a_2, b_1 + b_2)
    \end{align}
    where $a_1,a_2 \in A, b_1,b_2 \in B$. 
\end{definition}
It can be seen that if $A,B$ is Abelian, then $A \oplus B$ is Abelian. The above construction holds for more composing groups, that is, we can build ``larger'' groups, e.g., $ A_1\oplus A_2\oplus A_3 \oplus \cdots \oplus A_n$.  \\

Given a group $\mathcal{G}$, assumed to be Abelian for simplicity, with group operation $*$. Let $x \in \mathcal{G}$, $n \in \mathbb{Z}$ be some integers, and define $nx \equiv x*x*x*\cdots *x$ ($n$ times). Then $\mathcal{G}$ is said to be finitely generated if every element of $\mathcal{G}$ can be expressed as $ \sum n_i x_i$ where $n_i \in \mathbb{Z}$ and $x_i \in \mathcal{G}$. In this case, $\{x_i\}$ is said to be the generator of $\mathcal{G}$. We point out the following examples to illustrate this definition of generator:
\begin{itemize}
    \item The Abelian group $\mathbb{Z}^2$ is generated by $(1,0)$ and $(0,1)$. 
    \item The group $\mathbb{Z}_6 \equiv \{ 0,1,2,3,4,5\}$ is generated by $1$.  
\end{itemize}

At the heart of abstract algebra is the following theorem, popularly known as the structure theorem: 
\begin{theorem}[structure theorem for Finitely Generated Abelian Groups]
\label{thm: structuredtheorem}
    Every finitely generated Abelian group $\mathcal{G}$ can be decomposed as:
    \begin{align}
        \mathcal{G} \cong \Zbb^r \oplus \Zbb_{d_1} \oplus \Zbb_{d_2} \oplus \cdots \oplus \Zbb_{d_m}
    \end{align}
    where:
    \begin{enumerate}
        \item $r$ = rank, which is the number of independent ``infinite'' directions. 
        \item Each $\Zbb_{d_i} $ is a finite cyclic group of size $d_i$. 
        \item It holds that $d_1 | d_2 | \cdots | d_m$ where $d_i | d_{i+1}$ denotes the divisibility. 
    \end{enumerate}
     The first part $\Zbb^r$ is called \textit{free} part and $\Zbb_{d_1} \oplus \Zbb_{d_2} \oplus \cdots \oplus \Zbb_{d_m} $ is called \textit{torsion} part.
\end{theorem}
Aside from group, there are other algebraic objects, including ring and field. 
\begin{definition}[Ring]
    A ring $R$ is a set equipped with two operations: 
    \begin{enumerate}
        \item Addition (+), making $(R,+)$ an Abelian group. 
        \item Multiplication $(*)$, making $(R, *)$ a semingroup with an associativity $(a*b)*c = a*(b*c)$. 
        \item Distributivity: for all $a,b,c \in R$, it holds that:
        \begin{align}
            a*(b+c) = a*b + a*c \\
            (a+b)*c = a*c + b*c
        \end{align}
    \end{enumerate}
    A ring might not have a multiplicative identity. If it does, denote the identity as $1_R$ (which is not the identity element of the Abelian group under (+) operation), then we call $R$ a ring with unity. If the multiplicative operation commutes, for example, $a*b = b*a$, then the ring $R$ is called a commutative ring (possibly with unity). \\

    \noindent
    \textbf{Example:} $\Zbb$ (the integers) is a commutative ring with unity.
\end{definition}

\begin{definition}[Field]
    A field is a commutative ring $R$ with unity, where it holds that every element has a multiplicative inverse, e.g., for any $a \in R$, there is another $b$ such that $a*b = 1_R$. \\
    
    \noindent
    \textbf{Examples:} $\mathbb{Q}, \mathbb{R},\mathbb{C}$ are infinite fields. $\mathbb{Z}_p$ with $p$ prime is a field. 
\end{definition}

\begin{definition}[Characteristic of a field]
A characteristic of a field is the smallest positive number of times you add the multiplicative identity 1 to itself to reach the additive identity 0.
\end{definition}
If adding $1$ repeatedly never results in $0$, then the field is said to have characteristic zero. For any field, the characteristic of a field is always either $0$ or a prime number. \\

In the above, we have defined what it means to be a group and an Abelian group. It can be seen that a vector space is an Abelian group, with vector addition $(+)$ as a group operation. A vector space is usually defined over some field $\mathbb{F}$, and a vector space is also equipped with a scalar multiplication (with a few axioms associated). A generalization of vector space is called a \textit{module}, which is defined over a ring instead of a field. 
\begin{definition}[Module]
    Let $R$ be some ring. An $R$-module $M$ is:
    \begin{enumerate}
        \item An Abelian group $(M,+)$.
        \item Equipped with scalar multiplication $R \times M \longrightarrow M$, written $(r,m) \longrightarrow rm$ satisfying: 
        \begin{itemize}
            \item $r(m_1+ m_2) = r m_1 + rm_2$. 
            \item $(r_1+r_2) m = r_1  m + r_2 m  $. 
            \item $(r_1 r_2) m = r_1( r_2 m)$. 
            \item $1_R m = m$ (where $1_R$ denotes the identity element in the ring $R$). 
        \end{itemize}
    \end{enumerate}
\end{definition}
We recall a basis of a vector space in which an arbitrary element of some space can be written as a linear combination of the elements in the basis. For a set $S$, the free module $F(S)$ is made of (formal) linear combination of elements of $S$ with coefficients from a ring $R$. In linear algebra, given two vector spaces, one can construct the tensor product of them. In the setting of module, two modules can be tensor producted. Yet, the rule is different.
\begin{definition}[Tensor product of module]
\label{def: tensorproductofmodule}
    Given two $R$-modules $M,N$, the tensor product $M \otimes_R N$ is defined as:
    \begin{align}
        M \otimes_R N = F / \text{(bilinear relation)}
    \end{align}
    where $F$ is the free module generated by the pair $(m,n)$ and the bilinear relation is: 
    \begin{itemize}
        \item $(m_1+ m_2, n) = (m_1,n) + (m_2,n)$. 
        \item $(m,n_1+ n_2) = (m, n_1) + (m,n_2)$. 
        \item $(rm,n) = r (m,n) =(m, rn)$. 
    \end{itemize}
   The following properties are imposed within tensor product: 
    \begin{enumerate}
        \item $(m_1+ m_2) \otimes n = m_1 \otimes n + m_2 \otimes n$. 
        \item $ m \otimes (n_1 + n_2) = m \otimes n_1 + m \otimes n_2$. 
        \item $(rm)\otimes n = m \otimes (rn)$. 
    \end{enumerate}
\end{definition}
It can be seen that any Abelian group $\mathcal{G}$, with group operation $(*)$, is a $\Zbb$-module, by defining the action $n x = x * x * x *\cdots *x $ (n times) for $n \in \Zbb$ and $x \in \mathcal{G}$. We point out a few properties regarding computing tensor product between $\Zbb$-modules and other module:
\begin{lemma}[Tensor product of $\Zbb$-modules]
\label{lemma: propertytensorproductmodule}
Let $\Zbb$ be the set of integers. Then it holds that:
    \begin{itemize}
    \item $\Zbb^r \otimes_{\Zbb} N \cong N^r$ for arbitrary $\Zbb$-module $N$.
    \item $\Zbb_m \otimes_{\Zbb} \Zbb_n \cong \Zbb_{\rm gcd (m,n)} $.
\item $\Zbb_m \otimes_{\Zbb} Q = 0$ for $Q$ being a field of characteristic zero. 
\item $\Zbb^r \otimes_{\Zbb} Q \cong Q^r $
\end{itemize}
\end{lemma}
For illustration, we provide the following examples. \\

\noindent
\textbf{Example 1.} Let $M = \Zbb^2 \oplus \Zbb_4$, and $N = \Zbb_6$ be $\Zbb$-modules. Then we have the following. 
\begin{align}
    M \otimes_{\Zbb} N &= (\Zbb^2 \oplus \Zbb_4) \otimes_{\Zbb} \Zbb_6 \\
    &= ( \Zbb^2 \otimes_{\Zbb} \Zbb_6) \oplus (\Zbb_4  \otimes_{\Zbb} \Zbb_6 )\\
    &= \Zbb_6^2 \oplus \Zbb_2
\end{align}

\noindent
\textbf{Example 2.} Let $M = \Zbb^2 \oplus \Zbb_4$ and $N = Q$. Then we have: 
\begin{align}
    M \otimes_{\Zbb} N &= ( \Zbb^2 \oplus \Zbb_4 ) \otimes_{\Zbb} Q \\
    &= ( \Zbb^2 \otimes_{\Zbb} Q) \oplus (\Zbb_4 \otimes_{\Zbb} Q) \\
    &= 0 
\end{align}
From the above example, we see that whenever a $\Zbb$-module is tensor producted with a field (of characteristic zero), then the resulting module is zero. Subsequently, we will point out that this accounts for the fact that homology built on real coefficients miss the torsion part. 

\section{Algebraic Topology}
\label{sec: reviewofalgebraictopology}

This appendix provides a concise overview of the fundamental concepts in algebraic topology relevant to our work. We primarily follow the exposition of Nakahara~\cite{nakahara2018geometry, hatcher2005algebraic}, to which we refer interested readers for comprehensive treatment of the subject.

\begin{definition}[Simplex]
\label{def:simplex}
Let $p_0, p_1, \ldots, p_r \in \mathbb{R}^m$ be geometrically independent points where $m \geq r$. The $r$-simplex $\sigma_r = [p_0, p_1, \ldots, p_r]$ is defined as:
\begin{equation}
\sigma_r = \left\{ x \in \mathbb{R}^m : x = \sum_{i=0}^r c_i p_i, \quad c_i \geq 0, \quad \sum_{i=0}^r c_i = 1 \right\},
\end{equation}
where the coefficients $(c_0, c_1, \ldots, c_r)$ are called the barycentric coordinates of $x$.
\end{definition}

Geometrically, a 0-simplex $[p_0]$ represents a point, a 1-simplex $[p_0, p_1]$ represents a line segment, a 2-simplex $[p_0, p_1, p_2]$ represents a triangle, and higher-dimensional simplices generalize this pattern to higher dimensions.

\begin{center}
\begin{tikzpicture}[scale = 1.7, every node/.style={font=\small}]
    \filldraw (-5,-1) circle (1pt);
    \node[above left] at  (-5,-1) {$v_0$}; 
    \node[below right] at  (-5.5,-1) {$0$-simplex}; 
    
    \draw (-3.5,-1) -- (-2.5,-1);
    \filldraw (-3.5,-1) circle (1pt);
    \filldraw (-2.5,-1) circle (1pt);
    \node[above] at (-3.5,-1) {$v_0$};
    \node[above] at (-2.5,-1) {$v_1$}; 
    \node[below] at (-3.0,-1) {$1$-simplex}; 
    
    \coordinate (v0) at (-1,-0.5); 
    \coordinate (v1) at (-1.6, -1.5); 
    \coordinate (v2) at (-0.4, -1.5); 
    \draw[fill=blue!10] (v0)--(v1)--(v2)--cycle; 
    \filldraw (v0) circle (1pt);
    \filldraw (v1) circle (1pt);
    \filldraw (v2) circle (1pt);
    \node[above] at (v0) {$v_0$};
    \node[below left] at (v1) {$v_1$};
    \node[below right] at (v2) {$v_2$};
    \node[below] at (-1.0, -1.8) {$2$-simplex}; 
    
    \coordinate (P0) at (1.5,-0.4);
    \coordinate (P1) at (0.7,-1.4);
    \coordinate (P2) at (1.9,-1.6);
    \coordinate (P3) at (1.9,-1.0);
    \draw[fill=blue!10] (P0) -- (P1) -- (P2) -- cycle; 
    \draw[thick] (P0) -- (P1) -- (P2) -- (P0);
    \draw[fill=blue!10,thick] (P0) -- (P3) -- (P2);
    \draw[dashed] (P1) -- (P3);
    \filldraw (P0) circle (1pt);
    \filldraw (P1) circle (1pt);
    \filldraw (P2) circle (1pt);
    \filldraw (P3) circle (1pt);
    \node[above] at (P0) {$v_0$};
    \node[below left] at (P1) {$v_1$};
    \node[below right] at (P2) {$v_2$};
    \node[right] at (P3) {$v_3$}; 
    \node[below] at (1.3, -1.8) {$3$-simplex}; 
\end{tikzpicture}
\end{center}

An $r$-simplex can be assigned an orientation. For instance, the 1-simplex $[p_0, p_1]$ has orientation $p_0 \to p_1$, which differs from $[p_1, p_0]$. Throughout this work, we adopt the convention that for an $r$-simplex $[p_0, p_1, \ldots, p_r]$, the indices are ordered from low to high, indicating the canonical orientation.

\begin{definition}[Face and simplicial complex]
For an $r$-simplex $[p_0, p_1, \ldots, p_r]$, any $(s+1)$-subset of its vertices defines an $s$-face $\sigma_s$ where $s \leq r$. A simplicial complex $K$ is a finite collection of simplices satisfying:
\begin{enumerate}
\item Every face of a simplex in $K$ is also in $K$
\item The intersection of any two simplices in $K$ is either empty or a common face of both simplices
\end{enumerate}
The dimension of $K$ is $\dim(K) = \max\{r : \sigma_r \in K\}$.
\end{definition}

\begin{definition}[Chain group over real field]
\label{def: chaingroup}
Let $K$ be an $n$-dimensional simplicial complex. The $r$-th chain group $C_r^K$ is the free abelian group generated by the oriented $r$-simplices of $K$. For $r > \dim(K)$, we define $C_r^K = 0$. Formally, let $S_r^K = \{\sigma_{r,1}, \sigma_{r,2}, \ldots, \sigma_{r,|S_r^K|}\}$ denote the set of $r$-simplices in $K$. An $r$-chain is an element of the form:
\begin{equation}
c_r = \sum_{i=1}^{|S_r^K|} c_i \sigma_{r,i}
\end{equation}
where $c_i \in \mathbb{R}$ are real coefficients. The group operation is defined by:
\begin{equation}
c_r^{(1)} + c_r^{(2)} = \sum_{i=1}^{|S_r^K|} \left( c_i^{(1)} + c_i^{(2)} \right) \sigma_{r,i},
\end{equation}
making $C_r^K$ a free abelian group of rank $|S_r^K|$.
\end{definition}
In fact, since the coefficients $\{ c_i \}$ belong to a field $\Rbb$, the group $C_r^K$ is also a vector space. In the following, we use the Abelian group/vector space interchangeably. The boundary operator $\partial_r : C_r^K \to C_{r-1}^K$ is fundamental to homological algebra.

\begin{definition}[Boundary operator]
For an $r$-simplex $[p_0, p_1, \ldots, p_r]$, the boundary operator is defined as:
\begin{equation}
\partial_r [p_0, p_1, \ldots, p_r] = \sum_{i=0}^r (-1)^i [p_0, p_1, \ldots, \hat{p_i}, \ldots, p_r],
\end{equation}
where $\hat{p_i}$ indicates that vertex $p_i$ is omitted, yielding an $(r{-}1)$-simplex.

For an $r$-chain $c_r = \sum_{i=1}^{|S_r^K|} c_i \sigma_{r,i}$, we extend linearly:
\begin{equation}
\partial_r c_r = \sum_{i=1}^{|S_r^K|} c_i \partial_r \sigma_{r,i}.
\end{equation}
\end{definition}

The boundary operators form a chain complex:
\begin{equation}
0 \xrightarrow{} C_n^K \xrightarrow{\partial_n} C_{n-1}^K \xrightarrow{\partial_{n-1}} \cdots \xrightarrow{\partial_1} C_0^K \xrightarrow{\partial_0} 0.
\end{equation}

As an illustration, we consider $\partial_2$ and its action on a $2$-simplex:
\begin{center}
    \begin{tikzpicture}[scale = 1.7]
    \coordinate (v0) at (-1,-0.5); 
    \coordinate (v1) at (-1.6, -1.5); 
    \coordinate (v2) at (-0.4, -1.5); 
    \draw[fill=blue!10] (v0)--(v1)--(v2)--cycle; 
    \filldraw (v0) circle (1pt);
    \filldraw (v1) circle (1pt);
    \filldraw (v2) circle (1pt);
    \node[above] at (v0) {$v_0$};
    \node[below left] at (v1) {$v_1$};
    \node[below right] at (v2) {$v_2$};
    \node at (-2.8, -1) {$\partial_2 $};
    \draw[->] (-2.6,-1) -- (-1.6,-1);
    \node at (-2.2, -0.8)  {action}; 
    \node at (0.5, -1) {=};
    \filldraw (2, -0.5) circle (1pt); 
    \node[above] at (2,-0.5) {$v_0$};
    \filldraw (1.4, -1.5) circle( 1pt);
    \node[below left] at (1.4,-1.5) {$v_1$};
    \draw (2,-0.5) -- (1.4,-1.5);
    \filldraw (3, -1.5) circle (1pt);
    \node[below left] at (3,-1.5) {$v_1$};
    \filldraw (4.2, -1.5) circle (1pt);
    \node[below right] at (4.2,-1.5) {$v_2$};
    \draw (3,-1.5) -- (4.2, -1.5); 
    \filldraw (5.2, -0.5) circle (1pt);
    \node[above] at (5.2,-0.5) {$v_0$};
    \filldraw (5.8, -1.5) circle (1pt);
    \node[below right] at (5.8,-1.5) {$v_2$};
    \draw (5.2,-0.5) -- (5.8, -1.5);
    \node at (2.5, -1) { $+$};
    \node at (4.5, -1) {$-$};
\end{tikzpicture} 
\end{center}
\begin{definition}[Cycles, boundaries, and homology]
An $r$-chain $c_r$ is called an $r$-cycle if $\partial_r c_r = 0$. The collection of all $r$-cycles forms the $r$-cycle group $Z_r := \ker(\partial_r)$. Conversely, an $r$-chain $c_r$ is called an $r$-boundary if there exists an $(r{+}1)$-chain $d_{r+1}$ such that $\partial_{r+1} d_{r+1} = c_r$. The set of all $r$-boundaries forms the $r$-boundary group $B_r := \textnormal{im}(\partial_{r+1})$.
\end{definition}
A fundamental property of boundary operators is that $\partial_r \circ \partial_{r+1} = 0$ for any $r =0,1,2,...,n-1$, i.e., $\partial^2 = 0$, which ensures that every boundary is also a cycle  $B_r \subseteq Z_r$. This inclusion allows us to define the $r$-th homology group/space as the quotient group/space 
\begin{align}
    H_r(K, \mathbb{R}) = Z_r / B_r
\end{align}
which captures the notion of cycles that are not boundaries.

\begin{definition}[Betti numbers]
The $r$-th Betti number of the simplicial complex $K$ is defined as:
\begin{equation}
\beta_r(K) = \textnormal{dim}(  H_r(K, \mathbb{R})) = \textnormal{dim}(Z_r) - \textnormal{dim}(B_r).
\end{equation}
\end{definition}
In the group-theoretic language, the above dimension, e.g., $\textnormal{dim}(H_r^K) $, is replaced by $\rm rank (H_r^K)$. As we mentioned, we use the notion of group/space interchangeably. For computational purposes, we can utilize the combinatorial Laplacian:
\begin{equation}
\Delta_r = \partial_{r+1} \partial_{r+1}^\dagger + \partial_r^\dagger \partial_r,
\end{equation}
where $\partial_r^\dagger$ denotes the adjoint of $\partial_r$. A fundamental result in algebraic topology establishes that:
\begin{equation}
H_r(K, \mathbb{R}) \cong \ker(\Delta_r),
\end{equation}
providing a direct method for computing Betti numbers via kernel dimension.

A central theorem in algebraic topology states that homology groups constitute topological invariants~\cite{hatcher2005algebraic}:

\begin{theorem}[Topological invariance of homology]
If two topological spaces $X$ and $Y$ are homeomorphic, then their homology groups are isomorphic: $H_r(X, \mathbb{R}) \cong H_r(Y, \mathbb{R})$ for all $r \geq 0$. Consequently, their Betti numbers are equal: $\beta_r(X) = \beta_r(Y)$.
\end{theorem}
This invariance property makes Betti numbers powerful tools for topological classification and forms the mathematical foundation for their applications in topological data analysis.

\section{Role of Coefficients}

\subsection{ Theory }
\label{sec: roleofcoefficients}
In the above, we have built the chain group/space by taking (formal) linear combination of simplexes, with coefficients coming from a field. More concretely, we have defined the $r$-chain as: 
\begin{equation}
c_r = \sum_{i=1}^{|S_r^K|} c_i \sigma_{r,i}
\end{equation}
where $\{c_i\} \in \Rbb$. Instead, if we choose $\{c_i\}$ from the set of integers $\Zbb$, then the resulting set $C_r^K \equiv \{c_r = \sum_{i=1}^{|S_r^K|} c_i \sigma_{r,i}  \} $ is not a vector space, but just an Abelian group. We emphasize the subtlety that a vector space is inherently an Abelian group, but an Abelian group is not necessarily a vector space (as the multiplication by a scalar might not be properly defined, and also the scalar does not necessarily come from a field).  We note that in what follows, we implicitly understand that we are considering the simplicial complex $K$, and thus we would ignore $K$ in the notation. In other words, $H_r(\mathbb{R})$ appear below are exactly $ H_r(K,\mathbb{R})$ that we defined in the previous appendix. 

Changing the coefficients results in the change in the algebraic structure of homology groups. Suppose that we begin with the coefficients $\{c_i\} \in \Zbb$, then the $k$-th homology group $H_r(\Zbb)$ is an Abelian group (where we specify $\Zbb$ to explicitly imply the choice of coefficients from $\Zbb$). According to Theorem \ref{thm: structuredtheorem}, $H_r(\Zbb)$ admits the following decomposition:
\begin{align}
    H_r(\Zbb) \cong  \rm Free \big( H_r(\Zbb) \big) \oplus \rm Torsion\big( H_r(\Zbb)  \big)
\end{align}
where the free part $\rm Free$ generally has the form $\Zbb^r$ in which $r$ is called the rank, or the Betti numbers. The torsion part $\rm Tor$ is generally of the form $\Zbb_{d_1} \oplus \Zbb_{d_2} \oplus \cdots $ with the divisibility relation $d_i | d_{i+1}$. We can see that the $\rm Tor$ part is the main difference between homology with integer coefficients and real coefficients. To be more specific, if we change the choice of coefficients $\{c_i\}$ from $\Zbb $ to $\Rbb$, then there is a so-called universal coefficient theorem, which relates the change in the structure of homology group (under corresponding coefficients' type):
\begin{align}
    H_r(\Rbb) \cong H_r(\Zbb) \otimes_{\Zbb} \Rbb  \oplus \rm Tor \left( H_{r-1} (\Zbb), \Rbb   \right)
\end{align}
where the last term $\rm Tor \left( H_{r-1} (\Zbb), \Rbb   \right) $ is the \textit{torsion product functor}, which is defined as:
\begin{definition}
    [Torsion product functor]
    Let $\Zbb$ be the set of integers, and $m$ be some integers. Then for a cyclic group $\Zbb_m$ and another $\Zbb$-module $R$, $\rm Tor (\Zbb_m,  R) \cong [ r \in R | \  m r = 0 ] $. In other words, it is the $m$-torsion subgroup of $R$. 
\end{definition}
We remark that this $\rm Tor$ is different from the $\rm Torsion$ that we used earlier. For convenience, we point out the following useful properties associated with $\rm Tor$, and refer the readers to standard textbook for derivation. 
\begin{lemma}[$\rm Tor$ of $\Zbb$-modules]
\label{lemma: torsionfunctor}
Let $\Zbb$ be the set of integers, $p$ be a prime number and $R$ be some $\Zbb$-module.
\begin{itemize}
\item $\rm Tor (\Zbb, R) = 0$
    \item  $  \rm Tor \left( \bigoplus_i \Zbb_{m_i} , R \right)  = \bigoplus_i \rm Tor \left( \Zbb_{m_i}, R\right)$.
    \item Let $\mathbb{F}_p \equiv \Zbb/p\Zbb$ be the finite field of size $p$. Then:
    \begin{align}
    \rm Tor \left( \Zbb_m, \mathbb{F}_p \right)  \cong \begin{cases}
        \mathbb{F}_p \text{\ if p $\mid$ n} \\
        0  \text{\ if p $\nmid$ n} \\
    \end{cases}
\end{align}
\item For arbitrary $m$ (not necessarily a prime), it holds that:
\begin{align}
    \rm Tor (\Zbb_n, \Zbb_m) \cong \Zbb_{\rm gcd(m,n)}
\end{align}
\item For $R = \mathbb{Q}$ -- a divisible group, $ \rm Tor \left( \Zbb_{m}, R\right) = 0$ for any finite $m$. 
\end{itemize}
\end{lemma}
\noindent
Thus, a direct use of the above lemma leads us to:
\begin{align}
     H_r(\Rbb) \cong H_r(\Zbb) \otimes_{\Zbb} \Rbb 
\end{align}
To proceed, we consider the term $H_r(\Zbb) \otimes_{\Zbb} \Rbb $, which is:
\begin{align}
    H_r(\Zbb) \otimes_{\Zbb} \Rbb &\cong \left( \rm Free \big( H_r(\Zbb) \big) \oplus \rm Torsion\big( H_r(\Zbb)  \big) \right) \otimes_{\Zbb} \Rbb \\
    &= \left(  \rm Free \big( H_r(\Zbb) \big) \otimes_{\Zbb} \Rbb \right) \oplus \left(\rm Torsion\big( H_r(\Zbb)  \big) \otimes_{\Zbb} \Rbb  \right) 
\end{align}
We recall the properties from Lemma \ref{lemma: propertytensorproductmodule} that, for any integers $m$, $\Zbb_m \otimes_{\Zbb} Q = 0$ for any field $Q$ of characteristic zero. At the same time, the torsion part $ \rm Torsion\big( H_r(\Zbb)  \big)$ is $\cong \Zbb_{d_1} \oplus \Zbb_{d_2} \oplus \cdots$, which results in 
\begin{align}
    \rm Torsion\big( H_r(\Zbb)  \big) \otimes_{\Zbb} \Rbb    = 0
\end{align}
For the first part $  \rm Free \big( H_r(\Zbb) \big) \otimes_{\Zbb} \Rbb $, we can use the first property of Lemma \ref{lemma: propertytensorproductmodule}, and also $\rm Free \big( H_r(\Zbb) \big) \cong \Zbb^r $, which results in:
\begin{align}
      \rm Free \big( H_r(\Zbb) \big) \otimes_{\Zbb} \Rbb  \cong \Rbb^r
\end{align}
which implies that 
\begin{align}
    H_r(\Rbb) \cong \Rbb^r
\end{align}
suggesting that $H_r(\Rbb)$ behaves like a vector space. Thus, the above deduction reveals how the choice of coefficients influences the algebraic structure of homology, and that the typical choice of coefficients from $\Rbb$ or $\Zbb_2$ may hide the torsion part.

\subsection{A few examples}
\label{sec: afewexamples}
To make the role of torsion clearer, we provide a few concrete examples showing how the torsion is intrinsic to topological space. We particularly provide specific topological spaces with their corresponding homology groups. We organize these examples into two categories; one contains torsion-free topological space, while the other contains spaces with torsion.\\

\noindent
\textbf{Torsion-free:}
\begin{itemize}
    \item Sphere $S^n$: $H_0(\Zbb) \cong \Zbb$, $H_n(\Zbb) \cong \Zbb$. For any $0\leq m < n$, $H_m(\Zbb) \cong 0$. 
    \item Torus $T^n$: $H_r(T^n, \Zbb) \cong \Zbb^{ \binom{n}{k}}$. 
    \item Closed orientable surfaces (with genus $\geq 1$): $H_0(\Zbb) \cong \Zbb, H_1 (\Zbb)\cong \Zbb{2g}$, $H_2(\Zbb) \cong \Zbb$. 
    \item Complex projective space $\mathbb{CP}^n$: $H_{2k}(\Zbb) \cong \Zbb$ for $0 \leq k <n$; $H_{2k+1}(\Zbb) \cong 0$.
    \item Quaternionic projective space $\mathbb{HP}^n$: $H_{4k}(\Zbb) \cong \Zbb$ for $0 \leq k < n$. Otherwise 0. 
\end{itemize}

\noindent
\textbf{Torsion: } 
\begin{itemize}
    \item Real projective space $\mathbb{RP}^n$: $H_0(\Zbb) \cong \Zbb$. For $1 \leq i \leq n-1$, $H_i(\Zbb) \cong \Zbb_2$ if even $i$, and $H_i (\Zbb) \cong 0$ for odd $i$. $H_n(\Zbb) \cong \Zbb$ if $n$ is odd, and $H_n(\Zbb) \cong 0$ if $n$ is even. 
    \item Lens space $L(p,q)$ (3-manifolds): $ H_0(\Zbb) \cong \Zbb$, $H_1(\Zbb) \cong \Zbb_p$,  $H_2(\Zbb) \cong 0, H_3(\Zbb) \cong \Zbb$. 
    \item Klein bottle $K$: $H_0(\Zbb) \cong \Zbb, H_1(\Zbb) \cong \Zbb \oplus \Zbb_2, H_2(\Zbb) 0$. 
    \item Non-orientable close surface $N_g$ (connected sum of $g$ copies of $\mathbb{RP}^2$):
    $H_0(\Zbb) \cong 0, H_1(\Zbb) \cong \Zbb^{g-1} \oplus \Zbb_2, H_2(\Zbb) = 0$. 
    \item Moorse space $M(\Zbb_m, n)$: $H_n(\Zbb) \cong \Zbb_m$. 
\end{itemize}

\section{Key insight}
\label{sec: keyinsight}
In the following, we describe the insight that underlies our main algorithm to detect the emergence of torsion of a given simplicial complex, denoted $K$. Throughout the following, we omit this simplex notation $K$ and naturally impose this setting. We recall from the previous section the universal coefficient theorem:
\begin{align}
    H_r(\Rbb) \cong H_r(\Zbb) \otimes_{\Zbb} \Rbb  \oplus \rm Tor \left( H_{r-1} (\Zbb) , \Rbb  \right)
\end{align}
which reflects the change in algebraic structure upon a change in the coefficients from $\Zbb $ to $\Rbb$. If instead of $\Rbb$, we choose another ring/field $Q$, then the formula still holds:
\begin{align}
    H_r(Q) \cong H_r(\Zbb) \otimes_{\Zbb} Q  \oplus \rm Tor \left( H_{r-1} (\Zbb), Q  \right)
\end{align}
According to Theorem \ref{thm: structuredtheorem}, $H_r(\Zbb)$ can be decomposed as:
\begin{align}
    H_r(\Zbb) \cong \Zbb^{\beta_r} \oplus \left( \bigoplus_i Z_{r_i} \right)
\end{align}
Then we have that the above formula is:
\begin{align}
    H_r(Q ) \cong  \left( \Zbb^{\beta_r} \otimes_{\Zbb} Q  \right) \oplus \left( \bigoplus_i  Z_{r_i} \otimes_{ \Zbb} Q \right) \oplus  \rm Tor \left( H_{r-1} (\Zbb), Q \right)
\end{align}
To make the above group become a vector space (since linear algebra is a more natural language in quantum computation), we choose $Q$ to be some field. For a reason that will be clear later, we choose $Q = \mathbb{F}_p \equiv \Zbb/p\Zbb$ for $p$ being a prime. Then for $Q= \mathbb{F}_p$,  the above equation can be written as:
\begin{align}
    H_r(\mathbb{F}_p) \cong  \left( \Zbb^{\beta_r} \otimes_{\Zbb} \mathbb{F}_p  \right) \oplus \left( \bigoplus_i  \Zbb_{r_i} \otimes_{ \Zbb} \mathbb{F}_p \right) \oplus  \rm Tor \left( H_{r-1} (\Zbb), \mathbb{F}_p  \right)
\end{align}
Again, by Theorem \ref{thm: structuredtheorem}, we have:
\begin{align}
    H_{r-1} (\Zbb) \cong  \Zbb^{\beta_q}  \oplus \left( \bigoplus_i Z_{q_i} \right)
\end{align}
Via Lemma \ref{lemma: torsionfunctor}, we have that:
\begin{align}
    \rm Tor\left( H_{r-1} (\Zbb) , \mathbb{F}_p \right) & \cong \rm Tor \left(\Zbb^{\beta_q}  \oplus \left( \bigoplus_i Z_{q_i} \right), \mathbb{F}_p \right)  \\
    & \cong \rm Tor \left( \left( \bigoplus_i Z_{q_i} \right), \mathbb{F}_p \right)  \\
    & \cong \bigoplus_{i, q_i \mid p} \mathbb{F}_p
\end{align} 
By Lemma \ref{lemma: tensorproduct}, and also the fact that $\mathbb{F}_p \equiv \Zbb_p$, we have that:
\begin{align}
     \left( \bigoplus_i  \Zbb_{k_i} \otimes_{ \Zbb} \mathbb{F}_p \right) &\cong \bigoplus_{i, k_i \mid p} \mathbb{F}_p \\
      \left( \Zbb^{\beta_r} \otimes_{\Zbb} \mathbb{F}_p  \right)  &\cong \mathbb{F}_p^{\beta_r}
\end{align}
Gathering everything, we have:
\begin{align}
    H_r(\mathbb{F}_p) \cong  \mathbb{F}_p^{\beta_r} \oplus \left(\bigoplus_{i, r_i \mid p} \mathbb{F}_p \right) \oplus \left(\bigoplus_{i, q_i \mid p} \mathbb{F}_p \right) 
\end{align}
We recall from the previous section that if we choose the field $Q \equiv \Rbb$, then:
\begin{align}
    H_r(\Rbb) \cong \Rbb^{\beta_r}
\end{align}
It means that, the rank of the $k$-th homology space over $\mathbb{F}_p$ field contains additional factors from the torsion part, i.e.,:
\begin{align}
\rm dim \left( H_r(\mathbb{F}_p) \right) = \beta_r + t_r(p) + t_{r-1}(p) 
\end{align}
where $t_r(p), t_{r-1}(p)$ is the sum of the cyclic summands (of $k$-th homology group/space and $(k-1)$-th homology group/space, respectively) of order divisible by a prime $p$. We recall from the previous section \ref{sec: roleofcoefficients} that: 
\begin{align}
    H_r(\Rbb) \cong \Rbb^{\beta_r} 
\end{align}
which implies 
\begin{align}
    \dim \left(  H_r(\Rbb)\right)  = \beta_r
\end{align}
Our strategy is built on this, as we proceed to estimate the dimension of $k$-th homology group over two different fields $\mathbb{F}_p$'s, and then compare them. If their dimensions are equal, then it indicates that there is no torsion. Otherwise, there is torsion. The problem then is, what value of $p$ should we choose ? We recall that $p$ is a prime, and it is of a known fact that arbitrary integers can be decomposed as products of different primes. Thus, in general, without knowing the order of summands in advance (which is apparent, because if we know the order already, then the torsion is already known), the only strategy is to try as many primes as possible, e.g., $p=2,3,5 ..., 7$, etc. 

\section{Some operations involving topological spaces }
\label{sec: tspaceoperation}

\paragraph{Suspension of a topological space $\Sigma$.} The suspension of $X$ is obtained by roughly taking the cylinder $X \times [0,1]$ and collapse the entire bottom $X \times \{0\}$ to one point and the top $X \times \{1\}$ to another point:
\begin{align}
    \Sigma X = \frac{X \times [0,1] }{ X \times \{0\} \sim N, X \times \{1\} \sim S}
\end{align}
For example: 
\begin{align}
    \Sigma S^0 \cong S^1 , \Sigma S^1 \cong S^2
\end{align}
Generally: 
\begin{align}
    \Sigma S^{n} \cong S^{n+1}
\end{align}

\paragraph{Wedge product $X \vee Y$.} Choose a basepoint $x_0 \in X, y_0 \in Y$, Then $X \vee Y$ is obtained by identifying just two points:
\begin{align}
    X \vee Y = \frac{X \cup P}{x_0 \sim p_0}
\end{align}
Essentially, $X \vee Y$ means $X,Y$ glued together at some point. 

\paragraph{Smash product $X \wedge Y$.} This is defined as follows:
\begin{align}
    X \wedge Y = \frac{X \times Y}{ X \vee Y}
\end{align}

\paragraph{Join of topological spaces $X$ and $Y$.} The join is defined as follows:
\begin{align}
    X*Y = \frac{X \times P \times [0,1]}{ (x,p_1,0) \sim (x,p_2,0) , (x_1,p,1) \sim (x_2,p,1) } 
\end{align}
 
\paragraph{Useful identities.} Take $X = S^m, Y = S^n$, then the smash product is:
\begin{align}
    X \wedge Y = S^m \wedge S^n \cong S^{m+n}
\end{align}
We also have:
\begin{align}
    \Sigma (X \wedge Y) = \Sigma (S^m \wedge S^n) \cong \Sigma ( S^{m+n} ) \cong S^{m+n+1} \\
    X *Y = S^m * S^n = S^{m+n+1}
\end{align}
Generally, we have the following relation, which can be found in standard text \cite{hatcher2002algebraic}:
\begin{align}
    X*Y \cong \Sigma (X \wedge Y)
\end{align}

\section{Kunneth formula for topology}
\label{sec: kunneth}
The general Kunneth formula for homology can be found in any standard literature, e.g., \cite{hatcher2002algebraic}. Here we directly quote from Theorem 3B.6 in Section 3B of \cite{hatcher2002algebraic} as follows:
\begin{theorem}[Kunneth formula for topology]
    Let $X$ and $Y$ be CW complexes and let $R$ be some principal ideal domain. Then, for every $n \geq 0$, there is a natural short exact sequence:
    \begin{align}
        0 \longrightarrow \bigoplus_{p+q = n} H_p(X,R) \otimes_R H_q(Y,R) \longrightarrow H_n(X\times Y ,R) \longrightarrow \bigoplus_{p+q = n-1} \text{Tor}_1^R \big( H_p(X,R), H_q(Y,R)  \big) \longrightarrow 0
    \end{align}
    In particular, this short exact sequence splits. 
\end{theorem}
As a corollary of the theorem above, since the sequence splits, there exists an isomorphism of $R$-modules:
\begin{align}
    H_n(X \times Y, R) \cong \left[ \bigoplus_{p+q = n} H_p(X,R) \otimes_R H_q(Y,R)  \right] \oplus \left[ \bigoplus_{p+q = n-1}\text{Tor}_1^R \big(  H_p(X,R), H_q(Y,R)  \big)  \right]
\end{align}
In the same text \cite{hatcher2002algebraic}, it was proved that the Kunnet formula above admits the following reduced homology form:
\begin{theorem}[Reduced Kunneth formula for topology]
\label{thm: kunneth}
    Let $X$ and $Y$ be based CW complexes, with the base-points taken to be $0$-cells, and let $R$ be a principal ideal domain. Then, for every $n$, there is a split short exact sequence:
    \begin{align}
        0 \longrightarrow \bigoplus_{p+q=n} \widetilde{H}_p(X,R) \otimes_R \widetilde{H}_q(Y,R) \longrightarrow \widetilde{H}_n(X \wedge Y, R) \longrightarrow \bigoplus_{p+q = n-1} \text{Tor}_1^R \big( \widetilde{H}_p(X,R), \widetilde{H}_q(Y,R)  \big) \longrightarrow 0
    \end{align}
\end{theorem}
For integral homology, we take $R = \Zbb$, so we have:
\begin{align}
    0 \longrightarrow \bigoplus_{p+q=n} \widetilde{H}_p(X,\Zbb) \otimes_\Zbb \widetilde{H}_q(Y, \Zbb) \longrightarrow \widetilde{H}_n(X \wedge Y, \Zbb) \longrightarrow \bigoplus_{p+q = n-1} \text{Tor}_1^\Zbb \big( \widetilde{H}_p(X,\Zbb), \widetilde{H}_q(Y,  \Zbb)  \big) \longrightarrow 0
\end{align}
Then we have:
\begin{align}
    \widetilde{H}_n(X \wedge Y, R)  \cong \left[ \bigoplus_{p+q=n} \widetilde{H}_p(X,\Zbb) \otimes_\Zbb \widetilde{H}_q(Y, \Zbb)  \right] \oplus \left[ \bigoplus_{p+q = n-1} \text{Tor}_1^\Zbb \big( \widetilde{H}_p(X,R), \widetilde{H}_q(Y,R)  \big)  \right]
\end{align}
In the previous appendix, we have pointed out that, as proved in \cite{hatcher2002algebraic, davis2012geometry}: 
\begin{align}
    X*Y \cong \Sigma (X \wedge Y)
\end{align}
So for reduced homology:
\begin{align}
    \widetilde{H}_n( X *Y) \cong  \widetilde{H}_{n} ( \Sigma (X \wedge Y)) \cong \widetilde{H}_{n-1} (X \wedge Y)
\end{align}
where the second isomorphism is also a well-known result in the field \cite{hatcher2002algebraic, davis2012geometry}. So, we have the following:
\begin{align}
     0 \longrightarrow \bigoplus_{p+q=n-1} \widetilde{H}_p(X,\Zbb) \otimes_\Zbb \widetilde{H}_q(Y, \Zbb) \longrightarrow \widetilde{H}_{n} (X *Y, \Zbb) \longrightarrow \bigoplus_{p+q = n-2} \text{Tor}_1^\Zbb \big( \widetilde{H}_p(X,\Zbb), \widetilde{H}_q(Y,  \Zbb)  \big) \longrightarrow 0
\end{align}
This sequence implies:
\begin{align}
    \widetilde{H}_n(X*Y, \Zbb) \cong \left[ \bigoplus_{p+q=n-1} \widetilde{H}_p(X,\Zbb) \otimes_\Zbb \widetilde{H}_q(Y, \Zbb)  \right] \oplus \left[ \bigoplus_{p+q = n-2} \text{Tor}_1^\Zbb \big( \widetilde{H}_p(X,R), \widetilde{H}_q(Y,R)  \big)  \right]
    \label{eqn: kunnethformula}
\end{align}
In our context, $X$ is exactly $K = \text{Cl }(G)$, and $Y$ is $P$, which is the flag complex triangulation of $\mathbb{R}\mathbb{P}^2$. By choosing $q=1, n = r+2, p = r$, we have:
\begin{align}
    \widetilde{H}_{r+2}(K*P, \Zbb) \cong \left[  \widetilde{H}_r(K,\Zbb) \otimes_\Zbb \widetilde{H}_1(P, \Zbb)  \right] \oplus \left[  \text{Tor}_1^\Zbb \big( \widetilde{H}_{r-1}(K ,R), \widetilde{H}_1(P,R)  \big)  \right]
\end{align}
Because $K$ is the wedge product of spheres, it is torsion-free, and thus, as a consequence of Lemma \ref{lemma: torsionfunctor}:
\begin{align}
     \text{Tor}_1^\Zbb \big( \widetilde{H}_r(K ,R), \widetilde{H}_1(P,R)  \big)  = 0
\end{align}
So we arrive at what we mentioned in Eqn.~\ref{eqn: kunneth}:
\begin{align}
    \widetilde{H}_{r+2} (K*P,\Zbb) \cong \widetilde{H}_r(K,\Zbb)  \otimes_\Zbb \widetilde{H}_1(P,\Zbb)
\end{align}

\section{Classical Randomized Algorithm for Rank Estimation Over finite field}
\label{sec: randomizedrankestimation}
In this appendix, we outline and discuss the probabilistic method introduced in \cite{kaltofen1991wiedemann, eberly2017black} to find the rank of a matrix over a finite field. We then analyze a variant of this algorithm where we use pseudorandomness instead of randomness.

\subsection{Algorithm for low-rank estimation}
The randomized algorithm \cite{kaltofen1991wiedemann, eberly2017black} for finding the rank of a matrix over finite field $\mathbb{F}_p$ proceeds as follows. 

\begin{method}[Randomized Rank Estimation Over $ \mathbb{F}_p$ \cite{kaltofen1991wiedemann, eberly2017black} ]
\label{algo: rankestimatingfinitefield}
    Let $A \in \mathbb{F}_p^{m \times n}$ be a matrix of size $m \times n$ over the field $\mathbb{F}_p$. Let $R$ denotes the rank of $A$, which is assumed to be smaller than a known threshold, $R \leq \alpha$.
\end{method}
\begin{enumerate}
    \item Draw $t$ random vectors $V_1,V_2,...,V_t \in \mathbb{F}_p^n$. Define $V$ be a matrix of size $n \times t$, with $[V_1,V_2,...,V_t] $ as columns.
    \item Draw $s$ random vectors $U_1,U_2,..., U_s \in \mathbb{F}_p^m$. Define $U$ be a matrix of size $s \times m$, with $ [U_1^T,U_2^T,..., U_s^T]$ as rows.
    \item Define a matrix $M$ of size $s \times t$ with entries:
    \begin{align}
        M_{ij} = U_i^T A V_j  \text{ \ (mod p)}
    \end{align}
    Or equivalently
    \begin{align}
        M = U A V
    \end{align}
    \item Find rank $\hat{R}$ of $M$ (via Gaussian elimination mod p). 
    
    \noindent
    \textbf{Output:} $\hat{R}$ is a one-sided estimator for the rank $R$ of $A$. 
    
    \noindent
    \textbf{Guarantee: } For $s,t  \geq \alpha + \lceil \log_p \frac{2}{\Delta}\rceil$ (in fact, $s$ can be chosen to be equal to $t$ for simplicity) and the entries of $V,U$ are random with i.i.d. uniform over $\mathbb{F}_p$, the above estimation of $\hat{r}$ succeeds with probability:
    \begin{align}
        \rm Prob [ \hat{R} = R] \geq  1-\Delta
    \end{align}
\end{enumerate}
\noindent
In what follows, we give a detailed explanation for the algorithm above, thus justifying its correctness. 

\subsection{ Proof of performance guarantee of Algorithm \ref{algo: rankestimatingfinitefield}}
Let $R$ be the rank of $A$, then we write the factorization $A = BC$ where $ B \in \mathbb{F}_p^{m \times R}, C \in \mathbb{F}^{R \times n}$, and $\rm rank(B) = \rm rank (C) = R$. So we have $UAV = (UB) (CV)$. So, if $\text{rank } (UB), \text{rank } (CV)$ is equal to $R$, then it implies that $ \text{rank} (UB \cdot CV) = \text{rank} (UAV) = R$. Now we aim to estimate with what probability the rank of $UB, CV$ is equal to $R$. 

Since $U$ is a matrix of size $s \times m$ with entries drawn i.i.d. from $\mathbb{F}_p$, the matrix $UB$ also behaves exactly like a uniformly random $s \times R$ matrix over $\mathbb{F}_p$. According to \cite{salmond2014rank, blomer1997rank}, the probability that this $s \times R$ matrix has full column rank is:
\begin{align}
    \text{Prob [rank($UB$) = R])}  = \prod_{ i=0}^{R-1} (1- p^{i-s})
\end{align}
Similarly, the probability of $CV$ having full row rank is:
\begin{align}
     \text{Prob [rank($CV$) = R])}  = \prod_{ i=0}^{R-1} (1- p^{i-t})
\end{align}
If $U,V$ are independent, then the probability that both have full rank is:
\begin{align}
  \text{Prob [rank($UB$) =rank($CV$) =  R])}  = \prod_{ i=0}^{R-1} (1- p^{i-s})  \prod_{ i=0}^{R-1} (1- p^{i-t}) 
\end{align}
Therefore, the probability that $(UB)\cdot (CV)  = UAV$ has rank $R$ is the same. 

In an equivalent manner, we can analyze as follows. A simple union bound gives:
\begin{align}
   \text{Prob [rank($UB$) $<$ R])} \leq (p^R-1)p^{-s}  < p^{R-s}
\end{align}
and similarly:
\begin{align}
     \text{Prob [rank($CV$) $<$ R])} < p^{R-t}
\end{align}
Therefore, the probability that $\text{Prob [rank($UB$) = R])} , \text{Prob [rank($CV$) = R])}   $ are greater than:
\begin{align}
    1- p^{R-s} - p^{R-t}
\end{align}
which mans that: 
\begin{align}
  \text{Prob [rank( $UAV$)=  R])}   \geq 1- p^{R-s} - p^{R-t}
  \end{align}
For simplicity, since $R \leq \alpha$, choosing $s=t = \alpha + l$, then we have:
\begin{align}
     \text{Prob [rank( $UAV$)=  R])}   \geq 1-2p^{-l} 
\end{align}
If we desire this success probability to be $\geq 1-\Delta$, then we have $2p^{-l} = \Delta$, which leads to:
\begin{align}
    s = t = \alpha + l \geq \alpha + \lceil \log_p \frac{2}{\Delta}\rceil
\end{align}

\subsection{Algorithm \ref{algo: rankestimatingfinitefield} in the presence of pseudorandomness }
In Algorithm \ref{algo: rankestimatingfinitefield} above, the entries of $U,V$ are i.i.d. drawn from the finite field $\mathbb{F}_p$. Here we examine the performance guaranty of the algorithm in an alternate setting, where we do not have i.i.d. entries drawn from unifrom distribution, but rather from a $\epsilon$-biased distribution (we will define it later).

Consider the matrix $UB \in \mathbb{F}_p^{t \times r}$. $UB$ would lose its rank if there is some nonzero $z \in \mathbb{F}_p^r$ such that $UB z=  0$. Because $B$ has full column rank, then for $z \neq \bf{0}$, $Bz \neq \bf{0}$. Defining $y=Bz$, then the bad event become:
\begin{align}
    Uy = 0
\end{align}
Since the rows of $U$ are $u_1^T, ..., u_t^T$, the equation above implies $ u_1^T y = \cdots = u_t^T y =0$. For truly uniform, the value of $u_i^T y $ is uniformly distributed over $\mathbb{F}_p$. So we have:
\begin{align}
    \text{Pr} [ u_i^T y = 0] = \frac{1}{p}
\end{align}
If the rows of $U$ are independent, then:
\begin{align}
    \text{Pr} [Uy = 0] = \frac{1}{p^t}
\end{align}
Now consider $u \in \mathbb{F}_p^n$ but the entries of $u$ are not i.i.d. drawn from $\mathbb{F}_p^n$. Instead, it is drawn from the $\epsilon$-biased distribution $\mathscr{D}$. More formally:
\begin{definition}[$\epsilon$-biased distribution]
\label{eqn: epsilonbiaseddistribution}
A distribution $\mathscr{D}$ is called $\epsilon$-biased distribution if elements $u_1,u_2,...,u_n \in \mathbb{F}_p$ drawn from $\mathscr{D}$ satisfy the following:
    \begin{align}
    \Big|  \text{Pr}_{u \sim \mathscr{D}} [ u^T y = 0 ]  - \frac{1}{p} \Big| \leq (1- \frac{1}{p})\epsilon
\end{align}
where $u= (u_1,u_2,...,u_n)^T$ and $y \neq \textbf{0}$ with its entries $\in \mathbb{F}_p$. Equivalently:
\begin{align}
    \text{Pr}_{u \sim \mathscr{D}} [ u^T y = 0] \leq \frac{1}{p} + (1- \frac{1}{p})\epsilon
\end{align}
\end{definition}
This distribution behaves almost like a uniformly random one. Next, suppose that $u_1, u_2, ..., u_t$ are independently drawn from this $\epsilon$-biased distribution. Then:
\begin{align}
    \text{Pr} [Uy =0 ] = \text{Pr} [u_1^T y = \cdots = u_t^T y ] = 0 \leq \Big( \frac{1 + (p-1)\epsilon}{p} \Big)^t 
\end{align}
Recall that $UB$ loses rank if $UBz = 0$ for $z \neq \textbf{0}$. There are $p^R-1$ nonzero vectors $z$ in $\mathbb{F}_p^r$. For each one, $y = Bz \neq \textbf{0}$. Using a union bound, we have:
\begin{align}
    \text{Pr} [\text{rank} (UB)  < R ] \leq (p^R-1)  \Big( \frac{1 + (p-1)\epsilon}{p} \Big)^s
\end{align}
In a similar manner, we can show that:
\begin{align}
    \text{Pr} [\text{rank} (CV)  < R ] \leq (p^R-1)  \Big( \frac{1 + (p-1)\epsilon}{p} \Big)^t 
\end{align}
Using the union bound again, we have the following:
\begin{align}
      \text{Pr} [\text{rank} (UB CV)  < R] \leq (p^R-1)  \Big( \frac{1 + (p-1)\epsilon}{p} \Big)^t  + (p^R-1)  \Big( \frac{1 + (p-1)\epsilon}{p} \Big)^s
\end{align}
Note that $BC =A$, then we have:
\begin{align}
     \text{Pr} [\text{rank} (UAV)  = R ] \geq 1- (p^R-1)  \Big( \frac{1 + (p-1)\epsilon}{p} \Big)^t - (p^R -1)  \Big( \frac{1 + (p-1)\epsilon}{p} \Big)^s
\end{align}
For simplicity, $s$ can be chosen to be equal to $t$. To obtain the failure probability $\Delta$, we set:
\begin{align}
    2 (p^R-1)  \Big( \frac{1 + (p-1)\epsilon}{p} \Big)^t = \Delta
\end{align}
which implies that: 
\begin{align}
    s = t = \frac{\log \big( 2 (p^R-1) \frac{1}{\Delta} \big) }{ \log \big(  \frac{p}{1+(p-1)\epsilon} \big) } 
\end{align}

\section{Some notes on finite field $\mathbb{F}_p$ and finite field-extension $\mathbb{F}_q$ for $q= p^m$}
\label{sec: somenoteonfieldextension}
To understand how the field $\mathbb{F}_q$ is constructed from $\mathbb{F}_p$, we need the following recipes, which can be found in any standard text, e.g., \cite{lidl1997finite}.

\paragraph{Polynomial ring $\mathbb{F}_p[x]$.} Let $\mathbb{F}_p$ be the finite field of order $p$ (we recall that a field is a commutative ring with unity, and each element has a multiplicative inverse which also belongs to the same ring). Let $\mathbb{F}_p[x]$ be the set of polynomials with coefficients in $\mathbb{F}_p$. Then, under the multiplication rule of the elements of $\mathbb{F}_p$ and polynomials, this set can be shown to be a ring. For example, when $p=2$, then $\mathbb{F}_p[x]$ contains:
\begin{align}
    0, 1, x, x+1, x^2 +x +1, x^3 +x ,...
\end{align}

\paragraph{Irreducible polynomials.} Let $f(x) \in \mathbb{F}_p[x]$. Then $f(x)$ is called \textit{irreducible} if $f(x)$ cannot be factorized as $f(x) = h(x) g(x)$. For example, consider $p=2$, so $\mathbb{F}_p = (0,1)$. Then
\begin{align}
    x^2 + x +1 
\end{align}
is irreducible over $\mathbb{F}_2$. At the same time, $x^2 + x = x(x+1)$ and thus is not irreducible.

\paragraph{Quotient $\mathbb{F}_p[x]/f(x)$.} Let $f(x) \in \mathbb{F}_p[x]$ be some irreducible polynomial. The quotient $ \mathbb{F}_p[x]/f(x)$ means that polynomials are considered equivalent modulo $f(x)$, or equivalently, declaring $f(x) = 0$. For example, take $p=2$, and $f(x) = x^2 + x +1$ is an irreducible polynomial. Declaring
\begin{align}
    x^2 + x + 1 = 0 
\end{align}
we have $x^2 = -x - 1 = x+1$. So we now consider polynomials in $\mathbb{F}_2[x]$:
\begin{align}
    \begin{split}
        x^3 = x(x^2) = x(x+1) = x^2 + x = x+x+1 = 1\\
        x^4 = x^2 x^2 = (x+1)(x+1) = x^2 + 1 = x 
    \end{split}
\end{align}
Therefore, even though $\mathbb{F}_2[x]$ contains infinitely many polynomials, $\mathbb{F}_2[x]/f(x)$, after quotienting by $x^2 + x +1$, there are only four different equivalent classes:
\begin{align}
    \{ 0,1,x,x+1 \}
\end{align}
where we have abused a bit of notation, as $x$ here means that the equivalent class of $x$ (modulo $x^2+x+1$). It can be seen that $(0,1,x,x+1)$ actually forms a field of order $4$.  More generally, we have the following well-known property:
\begin{proposition}
    Let $f(x) \in \mathbb{F}_p[x]$ be some irreducible polynomial of degree $m$. Define $q= p^m$, then $\mathbb{F}_{q} := \mathbb{F}_p[x]/f(x)$ is a field of order $p^m$.
\end{proposition}
More generally, if $f(x)$ is an irreducible polynomial of degree $m$, then an element of $\mathbb{F}_q$ has the form:
\begin{align}
    a_0 + a_1 x + a_2 x^2 + \cdots + a_{m-1}x^{m-1}
\end{align}
where the coefficients $a_0, a_1,...., a_{m-1} \in \mathbb{F}_p$, which means $\mathbb{F}_q$ has $p^m$ elements.

\section{Sampling from $\epsilon$-biased distribution}
\label{sec: samplingfrombiaseddistribution}

Here we discuss how to efficiently generate $N$ elements $v_0,v_1,...,v_{N-1}$ from the $\epsilon$-biased distribution $\mathscr{D}$ as described above. Define $q= p^m$ and let $\mathbb{F}_q \equiv \mathbb{F}_{p^m}  $ be the finite field-extension of order $q$, as described in the previous appendix. The value of $m$ is chosen so that:
\begin{align}
    \frac{N-1}{p^m -1} \leq \epsilon
\end{align}
which implies that $m = \mathcal{O}\big(  \log_p  \frac{N}{\epsilon}  \big)$. 

Next, we choose two uniformly random elements $\alpha, \beta$ from $\mathbb{F}_q$ and define:
\begin{align}
    v_j = \Tr \big(  \beta \alpha^j \big), j = 0,1,2,...,N-1
\end{align}
\paragraph{Finite-field trace $\Tr(.)$ properties.} Here, the finite-field trace is defined as follows:
\begin{align}
    \Tr(z) =  z + z^p + z^{p^2} + \cdots + z^{p^{m-1}}
\end{align}
We have the following key properties.
\begin{proposition}
    $\Tr(z)$ is an element of $\mathbb{F}_p$. 
\end{proposition}
\textbf{Proof.}  This is a standard theorem in finite-field theory, which can be found in standard texts \cite{yuan2014further, lubeck2023standard, lidl1997finite}. Here, for completeness, we recapitulate the proof from \cite{yuan2014further}. Let:
\begin{align}
    T = \Tr(z) =  z + z^p + z^{p^2} + \cdots + z^{p^{m-1}}
\end{align}
Because the field $\mathbb{F}_q \equiv \mathbb{F}_{p^m}$ has characteristic $p$, which means:
\begin{align}
    (a+b)^p = a^p + b^p
\end{align}
Therefore:
\begin{align}
    T^p &= (  z + z^p + z^{p^2} + \cdots + z^{p^{m-1}})^p \\
    &= z^p + z^{p^2} + \cdots + z^{p^m}
\end{align}
There is a fundamental property of the finite field $\mathbb{F}_{p^m}$, that is $z^{p^m} = z$ for all $ z \in \mathbb{F}_{p^m}$. So:
\begin{align}
    T^p = z^p + z^{p^2} + \cdots + z^{p^{m-1}} + z 
\end{align}
which is exactly $T$. So we have:
\begin{align}
    T^p = T
\end{align}
The elements of $\mathbb{F}_{p^m}$ that satisfy the above equation is the roots of:
\begin{align}
    x^p - x = 0
\end{align}
We observe that all the elements of $\mathbb{F}_p =\{ 0,1,2, ..., p-1\}$ satisfy it. Thus, the root is $\mathbb{F}_p$, which implies that $T  \in \mathbb{F}_p$. Hence, for any $z \in \mathbb{F}_{p^m}$, it holds that $\Tr(z) \in \mathbb{F}_p$.  $\blacksquare$ \\

\paragraph{Examples.} The result above shows that $v_j$ belongs to the finite field $\mathbb{F}_p$. In the following, we give some concrete examples. Take
\begin{align}
    p=2, m= 2
\end{align}
which means that we consider the binary field $\mathbb{F}_2 = (0,1)$. Since the irreducible polynomial $f(x)$ has degree $m$, then the extension field $\mathbb{F}_4 := \mathbb{F}_2[x]/f(x)$ has $p^m = 2^2 = 4 $ elements. For simplicity, we choose $f(x) = x^2 + x +1$. As discussed in the previous appendix, this choice leads to four different equivalent classes:
\begin{align}
    \mathbb{F}_4 = \{  0,1,x, x+1 \}
\end{align}
where again we abuse the notation of $x$. Suppose that we choose:
\begin{align}
    \alpha = x, \beta = x + 1
\end{align}
The pseudorandom entries are:
\begin{align}
    v_j = \Tr ( \beta \alpha^j )
\end{align}
As $m=2$, we have that for any $z \in \mathbb{F}_4$, $\Tr(z) = z + z^2$. Now we calculate the trace of all the elements of $\mathbb{F}_4$:
\begin{align}
    \begin{split}
        \Tr(0) &= 0 \\
        \Tr(1) &= 1 + 1^2 = 0 \\
        \Tr(x) &= x + x^2 = x+ x + 1 = 1\\
        \Tr(x+1) &= (x+1) + (x+1)^2 = x+1 + x = 1
    \end{split}
\end{align}
Now we calculate $v_j$'s:
\begin{align}
    v_0 &= \Tr\big( \beta \alpha^0  \big) = \Tr\big( (x+1) 1  \big) = \Tr( x+1) = 1 \\
    v_1 &= \Tr\big( \beta \alpha^1   \big) = \Tr\big( (x+1)x  \big) = \Tr (1) = 0\\
    v_2 &= \Tr\big(  \beta \alpha^2   \big) = \Tr\big(  (x+1)x^2 \big) =\Tr(x) = 1 \\
    v_3 &= \Tr\big(  \beta \alpha^3 \big) = \Tr\big(  (x+1)x^3 \big) = \Tr(x+1) = 1
\end{align}
Therefore, all the values of $v_j$'s belong to $\mathbb{F}_2$. In addition, the trace function is also linear over the base field:
\begin{proposition}
    Let $z_1,z_2 \in \mathbb{F}_q$, then $\Tr(z_1+z_2) = \Tr(z_1) + \Tr(z_2)$.
\end{proposition}
A further scalar-linearity property is as follows:
\begin{proposition}
    Let $a \in \mathbb{F}_p, z \in \mathbb{F}_q$, then $\Tr(az) = a\Tr(z)$.
\end{proposition}
The proofs of these properties can be found in \cite{lidl1997finite}.

\paragraph{$v_j$'s form $\epsilon$-biased distribution.} 
The following proposition shows that all the $v_j$'s (for $j=0,1,...,N-1$) are drawn from the $\epsilon$-biased distribution.
\begin{proposition}
    $v_0,v_1,...,v_{N-1}$ belong to the $\epsilon$-biased distribution
\end{proposition}
\textbf{Proof.} This has been proven in \cite{alon1992simple} (in fact, in this work they consider binary field, $p=2$, but the generalization to arbitrary prime $p$ is straightforward), for which we recapitulate as follows. If we take any nonzero $y = (y_0,y_1,...,y_{N-1}) \in \mathbb{F}_p^N$ and consider the inner product (again we implicitly understand that we are working modulo $p$):
\begin{align}
    v^T y = \sum_{j=0}^{N-1} v_j y_j = \sum_{j=0}^{N-1} y_j \Tr \big(  \beta \alpha^j \big) = \Tr  \big(  \beta\sum_{j=0}^{N-1}  \alpha^j y_j \big)
\end{align}
Defining $P_y(\alpha) = \sum_{j=0}^{N-1} y_j \alpha^j$, which is a polynomial of degree at most $N-1$. If $P_y(\alpha) \neq 0$, then the multiplication by $P_y(\alpha)$ is a bijection of $\mathbb{F}_p$, so the uniformity of $\beta$ over $\mathbb{F}_q$ implies that $\beta P_y(\alpha)$ is uniform over $\mathbb{F}_q$. The trace map $\Tr: \mathbb{F}_q \longrightarrow \mathbb{F}_p$ is balanced, so:
\begin{align}
    \text{Pr} [ v^T y  =  0 | P_y(\alpha)  \neq 0] = \frac{1}{p}
\end{align}
When $P_y(\alpha) = 0$, then apparently $v^T y = 0$. Since $P_y(\alpha)$ has at most $N-1$ roots, we have:
\begin{align}
    \text{Pr} [ P_y(\alpha) =0] \leq \frac{N-1}{q-1} \leq \epsilon
\end{align}
Thus, for every nonzero direction:
\begin{align}
    \text{Pr} [ v^T y =0 ] \leq \frac{1}{p} + \epsilon
\end{align}
By a trivial rescaling $\epsilon = (1- \frac{1}{p}) \epsilon'$, then we have:
\begin{align}
     \text{Pr} [ v^T y =0 ] \leq \frac{1}{p} + \big(1- \frac{1}{p}\big) \epsilon'
\end{align}
which means that entries of $v$, $v_0,v_1,...,v_{N-1}$ are drawn from $\epsilon'$-biased distribution. $\blacksquare$

\section{Reversibly computing $v_j$}
\label{sec: reversiblycomputing}
In this section, we discuss how to implement an efficiently and reversibly classical circuit that computes $v_j := \Tr(\beta \alpha^j)$. The computation relies on a few recipes.

\paragraph{Multiplication of a field element by a field element is essentially a matrix multiplication.} Let $q=p^m$ as above and $c,z \in \mathbb{F}_{q}$ are two field elements. Then the multiplication $cz $ is, in fact, an $\mathbb{F}_p$-linear transformation. This can be seen by observing that an element of $\mathbb{F}_q$ has the form:
\begin{align}
      z_0 + z_1 x + z_2 x^2 + \cdots + z_{m-1}x^{m-1}
\end{align}
with $z_0, z_1, z_2, ..., z_{m-1} \in \mathbb{F}_p$. Then a multiplication by another field element $c$ would produce another field element of the same form, but the coefficients change. So if we treat each field element by a $m$-dimensional vector, then the resulting field element is another $m$-dimensional vector, and thus the transformation is essentially a multiplication by a $m \times m$ matrix (of course, eventually we need to take modulo $p$). In \cite{beauregard2003quantum}, the authors construct an explicitly reversible quantum circuit that performs the following transformation:
\begin{align}
    \ket{z} \longrightarrow \ket{cz}
\end{align}
where $\ket{z} := \ket{z_0, z_1, ..., z_{m-1}} \equiv \ket{z_0}\ket{z_1}\cdots\ket
z_{m-1}$ is the string of $m \log p$-bits corresponding to $z$, $c$ is classically known and $\ket{cz}$ is another string of $m \log p$-bits corresponding to $cz$. According to \cite{beauregard2003quantum}, the complexity of this circuit is $\mathcal{O}( m^2 \log^2 p)$.

\paragraph{Reversibly compute $\alpha^j$.} It can be seen that if $\alpha \in \mathbb{F}_{q}$ is provided (which means we know its coefficients), then by keeping using the method of \cite{beauregard2003quantum}, we can construct $\ket{\alpha^j}$. However, this procedure will take $j$ steps and therefore inefficient when $j$ is high, e.g., when $j = N$. In the following, we give a procedure that takes roughly $n=\log N$ steps. 

First, we note that $j$ is within the range $(0,N-1)$, so it can be represented by $\log N$ bits. We can write it as:
\begin{align}
    j = j_0 + 2^1 j_1 + 2^2 j_2 + \cdots 2^{n-1} j_{n-1}
\end{align}
where $j_0,j_1,...,j_{n-1} \in \{0,1\}$. Next, we have:
\begin{align}
    \alpha^j &= \alpha^{j_0 + 2^1 j_1 + 2^2 j_2 + \cdots 2^{n-1} j_{n-1} } \\
    &= \prod_{k=0}^{n-1} (\alpha^{2^k})^{j_k}
\end{align}
Since $\alpha$ is known, we can classically precompute all coefficients of $\alpha^{2^0}, \alpha^{2^1}, ..., \alpha^{2^{n-1}}$. This procedure will takes $n = \log N$ steps. We consider the following register:
\begin{align}
    \ket{j} \ket{1} = \ket{j_0 j_1 .... j_{n-1}} \ket{1}
\end{align}
where again $\ket{\alpha}$ contains the $m \log p$ bits corresponding to the coefficients of $\alpha$. Given $\alpha$ is known, this register can be efficiently prepared. Then starting from $j_0$, we use $\ket{j_0}$ as a controlled bit and use the method of \cite{beauregard2003quantum} to obtain the following transformation:
\begin{align}
    \ket{j_0 j_1 .... j_{n-1}} \ket{1} \longrightarrow \ket{j_0 j_1 .... j_{n-1}} \ket{ (\alpha^{2^0})^{j_0} 1} =  \ket{j_0 j_1 .... j_{n-1}} \ket{ (\alpha^{2^0})^{j_0}} 
\end{align}
Next, we use $\ket{j_1}$ as controlled bit and perform the transformation:
\begin{align}
    \ket{j_0 j_1 .... j_{n-1}} \ket{ (\alpha^{2^0})^{j_0}} \longrightarrow \ket{j_0 j_1 .... j_{n-1}} \ket{ (\alpha^{2^1})^{j_1}(\alpha^{2^0})^{j_0}} 
\end{align}
Continuing this process, we obtain the following:
\begin{align}
    \ket{j_0 j_1 .... j_{n-1}} \ket{(\alpha^{2^{j-1}})^{j_{n-1}} \cdots  (\alpha^{2^1})^{j_1}(\alpha^{2^0})^{j_0}} =  \ket{j_0 j_1 .... j_{n-1}}  \ket{\alpha^j} = \ket{j} \ket{\alpha^j}
\end{align}
Since each multiplication step using the method \cite{beauregard2003quantum} once, which has circuit complexity $\mathcal{O}(m^2 \log^2 p)$, so the transformation:
\begin{align}
    \ket{j} \ket{1} \longrightarrow \ket{j} \ket{\alpha^j}
\end{align}
has circuit complexity $\mathcal{O}(n m^2 \log^2 p)$. 

\paragraph{Reversible compute $\Tr( \beta \alpha^j)$.} Since $\beta \in \mathbb{F}_q$ is also known, we can use the same method of \cite{beauregard2003quantum} to construct a circuit of complexity $\mathcal{O}(m^2)$ that computes:
\begin{align}
    \ket{j} \ket{\alpha^j}  \longrightarrow \ket{j} \ket{\beta \alpha^j}
\end{align}
To compute $\Tr(\beta \alpha^j)$, we do the following steps. First, we express $\beta \alpha^j$ as:
\begin{align}
    a_0 + a_1 x + a_2 x^2 + \cdots + a_{m-1} x^{m-1}
\end{align}
Then we classically precompute $\Tr(x), \Tr(x^2), ..., \Tr(x^{m-1})$, for which the trace of $\beta \alpha^j$ can be expressed as:
\begin{align}
    \Tr(\beta\alpha^j) = a_0 + a_1 \Tr(x) + \cdots + a_{m-1} \Tr(x^{m-1}) 
\end{align}
modulo $p$. We remind that the register $\ket{\beta \alpha^j}$ contains $\ket{a_0, a_1...., a_{m-1}}$, so we can use the method of \cite{beauregard2003quantum} to obtain a quantum circuit that first implement:
\begin{align}
    \ket{j} \ket{\beta \alpha^j} \ket{0} :=  \ket{j} \ket{a_0, a_1, ..., a_{m-1}} \ket{0} \longrightarrow \ket{j} \ket{a_0, a_1\Tr(x), a_2\Tr(x^2), ...,  a_{m-1}\Tr(x^{m-1})} \ket{0}
\end{align}
In \cite{beauregard2003quantum}, the authors also construct a reversible quantum circuit that implements the following:
\begin{align}
    \ket{A} \ket{z} \longrightarrow \ket{A + z (\rm mod \ p)}\ket{z}
\end{align}
Using this circuit sequentially to our register, we obtain the sequential summation modulo $p$:
\begin{align}
    &\ket{j} \ket{a_0, a_1, ..., a_{m-1}} \ket{0} \\ 
    & \longrightarrow \ket{j} \ket{a_0, a_1\Tr(x), a_2\Tr(x^2), ...,  a_{m-1}\Tr(x^{m-1})} \ket{a_0} \\ 
    &\longrightarrow \ket{j} \ket{a_0, a_1\Tr(x), a_2\Tr(x^2), ...,  a_{m-1}\Tr(x^{m-1})} \ket{a_0 + a_1 \Tr(x) \text{ mod p}} \\
    & \longrightarrow \ket{j} \ket{a_0, a_1\Tr(x), a_2\Tr(x^2), ...,  a_{m-1}\Tr(x^{m-1})} \ket{a_0 + a_1 \Tr(x) + a_2 \Tr(x^2) \text{ mod p}} \\
    & \vdots \\
    & \longrightarrow \ket{j} \ket{a_0, a_1\Tr(x), a_2\Tr(x^2), ...,  a_{m-1}\Tr(x^{m-1})} \ket{a_0 + a_1 \Tr(x) + a_2 \Tr(x^2) + \cdots + a_{m-1} \Tr(x^{m-1}) \text{ mod p}}
\end{align}
The last register is exactly $\ket{\Tr (\beta \alpha^j)}$. The complexity of this procedure is thus $\mathcal{O}(m \log^2 p)$. Eventually, we uncompute everything, keeping only the final register that holds the value of $ \Tr (\beta \alpha^j)$.

To sum up, we have shown that the reversible quantum circuit that implements the following:
\begin{align}
    \ket{j} \ket{1}\longrightarrow \ket{j}\ket{\Tr (\beta \alpha^j)} \equiv \ket{j}\ket{v_j} 
\end{align}
can be efficiently constructed. 
The total complexity of this reversible circuit is $\mathcal{O}(nm^2 \log^2p + m \log^2 p) = \mathcal{O}(n m^2 \log^2 p)$. 

We remark that the procedure above is also the Proof of Lemma \ref{lemma: preparingstateepsilondistribution} by simply replacing $ n = \log N$. For a superposition, we can use this reversible circuit and it will enact on the superposition of states.


\section{Some graph-theoretic concepts and definitions}
\label{sec: graphtheory}
In this appendix, we summarize key concepts, notations, and definitions from graph theory that we have used in the main text. Let $G = (V,E)$ be the standard notation for the graph with the set of vertices $V$ and edges $E$ that connect a pair of vertices. 

\begin{definition}
    A graph $G$ is called \textit{chordal} if every cycle of length at least $4$ has a chord. A chord is an edge connecting two non-consecutive vertices of the cycle. 
\end{definition}
\begin{definition}
    The complement of a graph $G$, $\bar{G}$, has the same vertices $V$, but those edges in $\bar{G}$ connect vertices that are not connected in $G$. More formally:
    $$ E( \bar{G})  = \{  (u,v): u\neq v, (u,v) \neq E(G)  \} $$
\end{definition}

\begin{definition}
    A graph $G$ is co-chordal if its complement $\bar{G}$ is chordal. 
\end{definition}

\begin{definition}
    The independence complex of a graph $G$ is the simplicial complex whose simplexes are exactly the independent set of vertices of $G$.
\end{definition}
A very useful identity that can illuminate the definition above is:
$$ \text{Ind } (G) \equiv \text{Cl } (\bar{G})$$

\end{document}